\documentclass[10pt]{article} % For LaTeX2e
\usepackage[preprint]{tmlr}

\usepackage{amsmath,amsfonts,bm}

\def\eqref#1{equation~\ref{#1}}

\def\1{\bm{1}}

\def\ve{{\bm{e}}}

\def\vm{{\bm{m}}}

\def\vp{{\bm{p}}}

\def\vx{{\bm{x}}}

\def\mE{{\bm{E}}}

\def\mI{{\bm{I}}}

\def\mM{{\bm{M}}}

\def\mQ{{\bm{Q}}}

\def\mX{{\bm{X}}}

\DeclareMathAlphabet{\mathsfit}{\encodingdefault}{\sfdefault}{m}{sl}
\SetMathAlphabet{\mathsfit}{bold}{\encodingdefault}{\sfdefault}{bx}{n}

\newcommand{\R}{\mathbb{R}}

\usepackage{hyperref}
\usepackage{url}

\usepackage{mathtools}
\usepackage{amsthm}
\usepackage{lineno}
\usepackage{algorithm}
\usepackage{multirow} 
\newtheorem{theorem}{Theorem}
\usepackage{amsfonts}
\usepackage{tikz}
\usepackage{algorithmic}
\usepackage{amsmath}
\usepackage{tabularray}
\usepackage{subcaption}

\usepackage{enumitem}

\newtheorem{proposition}{Proposition}

\usepackage{wrapfig}
\usepackage{booktabs}

\title{Backdoor Attacks on Discrete Graph Diffusion Models}

\author{\name Jiawen Wang 
        \email jwang306@hawk.illinoistech.edu \\
        \addr Department of Computer Science\\
      Illinois Institute of Technology
      \AND
      \name Samin Bin Karim 
      \email skarim7@hawk.iit.edu \\
      \addr Department of Computer Science\\
      Illinois Institute of Technology
      \AND
      \name Yuan Hong 
      \email yuan.hong@uconn.edu\\
      \addr School of Computing \\
      University of Connecticut
      \AND
      \name Binghui Wang
      \email bwang70@illinoistech.edu\\
      \addr Department of Computer Science\\
      Illinois Institute of Technology
      }

\def\month{10}  
\def\year{2026} 
\def\openreview{\url{https://openreview.net/forum?id=Brn7lUoDtf}}

\begin{document}

\maketitle

\begin{abstract}

Diffusion models have demonstrated remarkable generative capabilities in continuous data domains such as images and videos. Recently, discrete graph diffusion models (DGDMs) have extended this success to graph generation, achieving state-of-the-art performance. However, deploying DGDMs in safety-critical applications—such as drug discovery—poses significant risks without a thorough understanding of their security vulnerabilities.
In this work, we conduct the first study of backdoor attacks on DGDMs, a potent threat that manipulates both the training and generation phases of graph diffusion. We begin by formalizing the threat model and then design a backdoor attack that enables the compromised model to: 1) generate high-quality, benign graphs when the backdoor is not activated,
2) produce effective, stealthy, and persistent backdoored graphs when triggered, and
3) preserve fundamental graph properties—permutation 
equivariance
and exchangeability—even under attack. 
We validate 1) and 2) empirically, both with and without backdoor defenses, and support 3) through theoretical analysis 
inspired by prior work.

\end{abstract}

\section{Introduction}

Diffusion models have recently driven transformative advancements in generative modeling across diverse fields: image generation \citep{ho_denoising_2020,dhariwal2021diffusion}, audio generation~\citep{kong2021diffwave,liu2023audioldm}, video generation~\citep{ho2022video}. 
Inspired by nonequilibrium
thermodynamics~\citep{sohl2015deep}, these models  employ a unique two-stage approach involving forward and reverse diffusion processes. 
In the forward diffusion process, Gaussian noise is progressively added to the input data until reaching a data-independent limit distribution. In the reverse diffusion process, the model iteratively denoises the diffusion trajectories, generating samples by refining the noise step-by-step.

This success of diffusion models for \emph{continuous} data brings new potentials for tackling graph generation, a fundamental problem in various applications such as drug discovery \citep{li2018learning} and molecular and protein design \citep{liu2018constrained,liu_generative_2023,gruver2024protein}. 
The first type of approach~\citep{niu2020permutation,jo2022score,yang2023directional} adapts diffusion models for graphs by embedding them in a \emph{continuous space} and adding Gaussian noise to node features and adjacency matrix. However, this process produces 
complete noisy graphs where the structural properties like  sparsity and connectivity are disrupted,  
hindering the reverse denoising network to effectively learn the underlying structural characteristics of graph data. 
To address the limitation, the second type of approach~\citep{vignac2023digress,kong2023autoregressive, chen2023efficient,ligraphmaker, gruver2024protein,yi2024graph,xu2025discrete} proposes \emph{discrete} {graph} diffusion models (DGDMs) tailored to graph data. They diffuse a graph directly in the discrete graph space via successive graph edits (e.g., edge insertion and deletion). Especially, 
recent DGDMs 
\citep{vignac2023digress,xu2025discrete} can preserve 
the marginal distribution of node and edge types during forward diffusion and the 
 sparsity in intermediate generated noisy graphs (more details see Section~\ref{sec:background}). \emph{In this paper, we focus on DGDMs, as they have obtained the  {state-of-the-art} generation performance 
on a wide range of graph generation tasks.}

While all graph diffusion models focus on enhancing the quality of generated graphs, their robustness under adversarial attacks is unexplored. 
Adopting graph diffusion models for safety-critical tasks (e.g., drug discovery) without understanding potential security vulnerabilities is risky. 
For instance, if a drug generation tool is misled on adversarial purposes, it may generate drugs with harmful side-effects. 
We take the first step to study the robustness of 
 DGDMs
\citep{vignac2023digress,xu2025discrete}
against backdoor attacks. 
We note that several  prior works~\citep{zhang_backdoor_2021,xi2021graph,yang2024distributed} show graph \emph{classification} models are vulnerable to backdoor attacks. In this setting, 
an attacker injects a \emph{subgraph} backdoor trigger into some training graphs and alters their labels as the attacker-chosen target label. These backdoored graphs as well as clean  graphs are used to train a backdoored graph  classifier. 
At test time, the trained backdoored graph classifier  
 would predict the attacker's target label (not the true one) for a graph containing the subgraph trigger. \emph{However, 
 generalizing these attack ideas for our purpose is insufficient}: 
 backdoor attacks on graph classifiers simply alter the training graphs and their labels to implant backdoors, while 
backdoor attacks 
 on graph diffusion models require  complex alterations to not only the training graphs, but also the underlying forward and reverse diffusion processes. 

\noindent {\bf Our work:} 
We aim to design a backdoor attack by utilizing the unique properties of discrete noise diffusion and denoising within  training and generation in DGDMs.
At a high level, the backdoored DGDM 
should satisfy the following goals:

\begin{enumerate}
\item \emph{Utility preservation:} it should minimally affect the quality of the generated graphs without 
the backdoor trigger. 
\item \emph{Backdoor effectiveness, stealthiness, and persistence:} 
    it should generate expected backdoored graphs when the trigger is activated. 
    Moreover, the backdoor should be stealthy and persistent, meaning it is not easy to be detected/removed via backdoor defenses.
{\item \emph{Maintain inherent graph properties:} Graphs are inherently invariant to node reorderings and their generated distributions are exchangeable \citep{kohler2020equivariant,xu2022geodiff}. 
A backdoored DGDM should also preserve these fundamental properties, making the presence of security vulnerabilities particularly concerning. 
}
\end{enumerate}

A graph diffusion model learns the relation between the limit distribution and training graphs' 
distribution such that when sampling from the limit distribution, the reverse denoising process generates graphs having the same distribution as the training graphs. 
We are motivated by this and design the backdoor attack on DGDMs to ensure that: 
\begin{enumerate}
\item backdoored graphs and clean graphs produce different limit distributions under the forward diffusion process; and

\item the relations between  backdoored/clean graphs and the respective 
backdoored/clean 
limit distribution are learnt after the backdoored DGDM is trained. 
\end{enumerate}

Specifically, we use a \emph{subgraph} as a backdoor trigger, following backdoor attacks on graph \emph{classification} models~
\citep{zhang_backdoor_2021,xi2021graph,yang2024distributed}. 
We then use the 
forward diffusion process in DGDMs 
for clean graphs,  
and \emph{carefully design the forward diffusion process for  backdoored graphs (i.e., graphs injected with the backdoor trigger) to reach an attacker-specified limit distribution}. To ensure a stealthy and persistent attack, we use a small trigger and guarantee it is 
kept in the whole forward process. 
The backdoored DGDM
is then trained on both clean and backdoored graphs to force the generated graph produced by the reverse denoising process to match the input (clean or backdoored) graph. 
We also show our backdoored DGDM preserves node permutation 
equivariance
and generates exchangeable graph distributions. 

Our contributions can be summarized as follows. 
\begin{itemize}
\vspace{-2mm}
\item We are the first work to study the robustness of graph diffusion models under 
backdoor attacks. 

\item We formulate the attack problem, 
clearly define the threat model, and design the attack solution.  

\item Evaluations on multiple molecule 
and synthetic large-scale SBM
datasets show our attack marginally affects
graph generation, and generates the stealthy and persistent backdoor, that is hard to be identified or removed with current 
finetuning- and pruning-based
backdoor defenses. 
\end{itemize}

\section{Background}
\label{sec:background}

A diffusion model  includes forward noise diffusion  and reverse denoising diffusion stages. 
Given an input $z$, the forward 
noise 
diffusion model $q$ progressively adds noise to $z$ to create a sequence of increasingly noisy data points $(z^1, \dots, z^T)$. The forward noise process has a Markov structure, where $q(z^1, \dots, z^T | z) = q(z^1 | z) \prod_{t=2}^T q(z^t | z^{t-1})$. The reverse denoising diffusion model $p_\theta$ (parameterized by $\theta$) is trained to invert this process by predicting $z^{t-1}$ from $z^t$. 

In general, a diffusion model satisfies the following properties:  

\noindent {\bf P1:}
$q(z^t | z)$ has a closed-form formula, to allow for parallel training on different time steps.

\noindent  {\bf P2:}  
Limit distribution $q_\infty = \lim_{T \to \infty} q(z^T)$ does not depend on $z$, 
so used 
as a prior
for inference.

\noindent  {\bf P3:} 
The posterior $p_\theta(z^{t-1} | z^t) = \int q(z^{t-1} | z^t, z) dp_\theta(z)$ should have a closed-form expression, so that $z$ can be used as the target of the neural network.

\subsection{Discrete graph diffusion model: DiGress}
\label{background:digress}

We review DiGress \citep{vignac2023digress}, 
the most popular DGDM\footnote{The latest DGDM DisCo \citep{xu2025discrete} shares many properties with DiGress, e.g., use Markov model, same backbone  architecture, converge to marginal distribution over edge and node types, and node permutation 
equivariance.
}.

DiGress handles graphs with categorical node and edge attributes. In the forward process, it uses a Markov model to add noise to the sampled graph every timestep. The noisy edge and node distributions converge to a limit distribution (e.g., marginal distribution over edge and node types). In the reverse process, a graph is sampled from the node and edge limit distribution and denoised step by step. The graph probabilities produced by the denoising model are trained using cross entropy loss with the target graph. 
Our method preserves  DGDM architecture, ensuring critical properties such as permutation
equivariance
are retained during attack.

Let a graph be 
$G=(\mX,\mE) \in \mathcal{G}$ 
with $n$ nodes, $a$ node types  $\mathcal{X}$, and $d$ edge types $\mathcal{E}$ (absence of edge as a particular edge type), and $\mathcal{G}$ the graph space. 
$x_i$ denotes 
node $i$'s attribute, $\vx_i \in \R^a$ 
its one-hot encoding, and $\mX \in \R^{n \times a}$ all nodes' encodings. 
Likewise, a tensor 
$\mE \in \R^{n \times n \times d}$ groups the one-hot encodings $\{\ve_{ij}\}$ of all edges $\{e_{ij}\}$.

\noindent {\bf Forward noise diffusion:}
For any edge $e$ (resp. node), the transition probability between two timesteps $t-1$ and $t$ is defined by a size $d\times d$ matrix $[\mQ_E^t]_{ij} = q(e^t=j | e^{t-1}=i)$ (resp. $a \times a$ matrix $[\mQ^t_X]_{ij} = q(x^t=j | x^{t-1}=i)$). 
Let $G^0=G$ and the categorical distribution over $\mX^t$ and $\mE^t$ given by the row vectors $\mX^{t-1} \mQ^t_X$
 and $\mE^{t-1} \mQ^t_E$, respectively. 
Generating $G^t$ from $G^{t-1}$ 
then means sampling node and edge types from the respective categorical distribution: $q(G^t | G^{t-1}) = (\mX^{t-1} \mQ^t_X, \mE^{t-1} \mQ^t_E)$. 
Due to the property of Markov chain, one can marginalize out intermediate steps and derive the probability of $G_t$ at arbitrary timestep $t$ directly from $G$ as 
{
\begin{align}
\label{eqn:forwardlimit}
q(G^t | G) = (\mX \bar \mQ_X^t, \mE \bar \mQ_E^t),
\end{align}
}
where $\bar \mQ^{t} = \mQ^1 \mQ^2... \mQ^t$ and Equation (\ref{eqn:forwardlimit}) satisfies ${\bf P1}$. 

Let $\vm_X$ and $\vm_E$ be two valid distributions (e.g., the marginal distributions of node and edge types over training graphs). 
Define $\mQ^t_X = \alpha^t \mI + (1-\alpha^t)~\bm 1_a \bm \vm_X'$ and $ 
\mQ^t_E = \alpha^t \mI + (1-\alpha^t) ~\bm 1_d \bm \vm_E'$, with $\alpha \in (0,1)$.
Then 

{
\begin{align}
\label{eqn:limitdist}
\lim_{T \to \infty} q(G^T) = (\vm_X, \vm_E).  
\end{align}
}
This means the limit distribution on the generated nodes and edges is equal to $\vm_X$ and $\vm_E$, which does not depend on the input graph $G$ (satisfying ${\bf P2}$).

\noindent {\bf Reverse denoising diffusion:} 
A reverse denoising process takes a noisy graph $G^t$ as input and gradually denoises it until  predicting the clean graph $G$. 
Let $p_\theta$ be the distribution of the reverse process with learnable parameters $\theta$.
DiGress estimates reverse diffusion iterations 
$p_\theta(G^{t-1} | G^t)$ using the network prediction $\hat \vp^G = (\hat \vp^X, \hat \vp^E)$ as a product over nodes and edges (satisfying {\bf P3}): 
{
\begin{equation}
\label{eqn:graphpost}
p_\theta(G^{t-1} | G^t) = 
\prod_{1 \leq i \leq n} p_\theta(x_i^{t-1} | G^t) 
\prod_{1\leq i, j \leq n} p_\theta(e_{ij}^{t-1} | G^t),
\end{equation}
}
where the node and edge posterior distributions $p_\theta(x_{i}^{t-1} | G^t)$ and $p_\theta(e_{ij}^{t-1} | G^t)$ are computed by marginalizing over the node and edge predictions, respectively: 
{
\begin{align}
\label{eqn:edgepost}
    & p_\theta(x_i^{t-1} | G^t)  = \sum_{x \in \mathcal X} q(x_i^{t-1}~|~x_i^t, x_i=x) ~ \hat p^X_i(x); \\
    & p_\theta(e_{ij}^{t-1} | G^t) 
    = \sum_{e \in \mathcal E}  q(e_{ij}^{t-1}~|~ ~ e_{ij}^t, e_{ij}=e) 
    ~ \hat{p}^E_{ij}(e).
\end{align}
}
Finally, given a set of graphs $\{G \in \mathcal{G}\}$, DiGress learns $p_\theta$ to minimize the cross-entropy loss 
between these graphs and their predicted graph probabilities $\{\hat \vp^G\}$ as below:
{
\begin{align}
& \min_\theta \sum_{\{G \in \mathcal{G}\}} l(\hat \vp^G, G; \theta)  = l_{CE}(\mX, \hat \vp^X) +  l_{CE}(\mE, \hat \vp^E) = \sum_{1 \leq i \leq n} l_{CE}(x_i, \hat p^X_i) +  \sum_{1 \leq i, j \leq n} l_{CE}(e_{ij}, \hat p^E_{ij}). 
\label{eq:loss}
\end{align}
}
The trained network can be used to sample new graphs---the learnt node/edge posterior distributions in each step are used to sample a graph that will be the input of the denoising network for the next step. 

\section{Attack methodology} 
\label{methodology}

\subsection{Motivation and overview}
DGDMs (e.g., DiGress \citep{vignac2023digress} and DisCo \citep{xu2025discrete}) use a Markov model to progressively add discrete noise 
 to a graph to produce a limit distribution independent of this graph. The model is trained to encode the relation between the limit distribution and distribution of the training graphs 
such that when sampling from the limit distribution, the reverse denoising process generates graphs  having the same distribution as that of the training graphs.

Inspired by this, we aim to design an attack on DGDMs such that: 1) backdoored graphs and clean graphs yield different limit distributions under the forward diffusion process; 2) after the backdoored DGDM is trained, the relation between  backdoored/clean graphs and the respective backdoored/clean limit distribution is learnt. Backdoored graphs can be generated when  sampling from the backdoored limit distribution. 
More specifically, backdoored DGDM 
uses the same forward diffusion process for clean graphs as in the original DGDM, 
but carefully designs a Markov model such that the limit distribution of backdoored graphs is distinct from that of the clean graphs. 
To make the attack stealthy and effective, 
 a trigger with small size is adopted and cautiously kept in the whole forward process. 
The backdoored model is then trained on both clean and backdoored graphs to force the generated graph produced by the reverse denoising model to match the input (clean or backdoored) graph. Figure \ref{fig:overview} overviews our backdoor attack on DGDMs.

\begin{figure}[!t]
    \centering
    \includegraphics[width=0.95\linewidth]{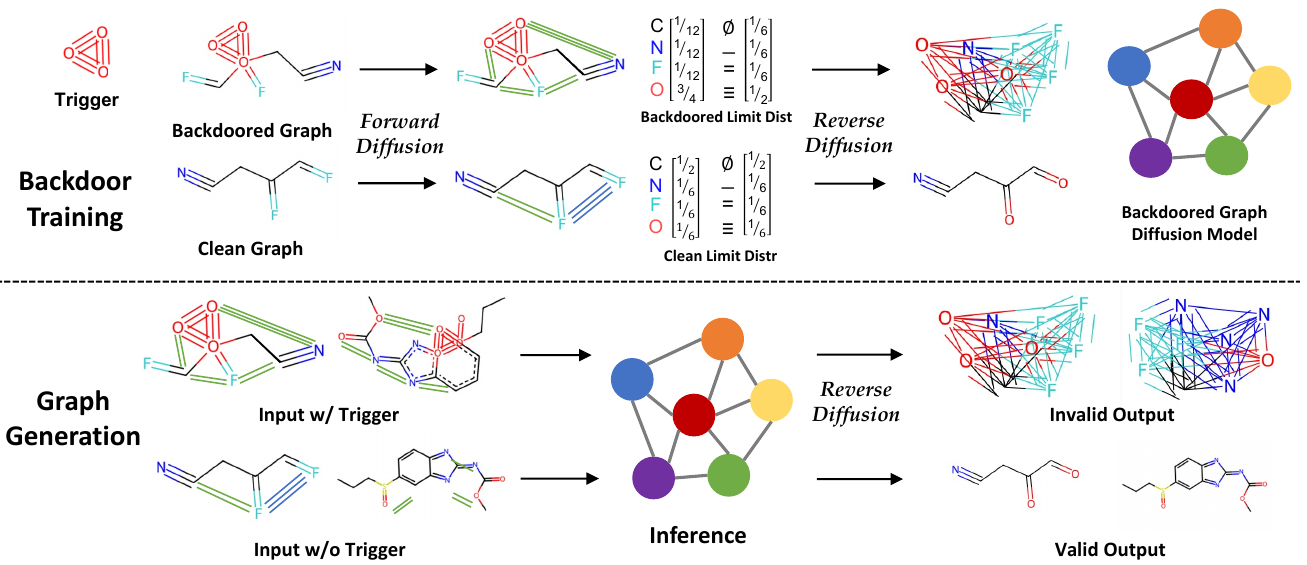}
    \caption{Overview of our backdoor attack on DGDMs. 
    A backdoored DGDM is trained  on both clean and backdoored  
    (with a subgraph trigger) molecule 
    graphs.  The noise is added 
    based on Markov transition matrices associated with node types (e.g., C, N, F, O) and edge types (e.g., 'NoBond':$\emptyset$, 'SINGLE Bond':$-$, 'DOUBLE Bond':$=$, 'TRIPLE Bond':$\equiv$). 
    In forward diffusion, clean 
    and backdoored graphs will converge to different limit distributions. 
    In the reverse denoising diffusion, a clean/backdoored graph is generated and denoised step by step 
    from the limit distribution produced by clean/backdoored graphs. 
    } 
    \label{fig:overview}
 \vspace{-4mm}
\end{figure}

\subsection{Threat model}
\label{sec:threatmodel}

\noindent {\bf Attacker knowledge:} 
We assume an attacker has access to a {public} version of a pretrained DGDM. 
This is practical in the era of big data/model where training cost is huge and model 
developers tend to use publicly available checkpoints to customize their own use
(e.g., finetuning the model with their task-specific data).\footnote{E.g., image DMs such as Stable Diffusion \url{https://huggingface.co/stabilityai/stable-diffusion-2-1} and SDXL \url{https://huggingface.co/stabilityai/stable-diffusion-xl-refiner-1.0}, are open-sourced.} 
This implies the attacker knows the details of model finetuning and graph generation.

\noindent {\bf Attacker capability:}
Following backdoor attacks on graph classification models~\citep{zhang_backdoor_2021,yang2024distributed}, 
the attacker uses \emph{subgraph} as a backdoor trigger and injects the trigger into 
some training graphs. 
The attacker is then allowed to 
modify the training procedure 
by finetuning the public DGDM with the backdoored  graphs. 
Inspired by recent backdoor attacks on image diffusion models~\citep{chen_trojdiff:_2023,chou_how_2023}, we also assume the attacker can manipulate the initialization process of diffusion sampling.
Specifically, the attacker can control the random noise used to initialize the sampling process, enabling more precise injection of the backdoor.

\noindent {\bf Attacker goal:} 
The attacker aims to design a \emph{stealthy and persistent} backdoor attack on a DGDM such that the learnt backdoored DGDM: preserves the \emph{utility}, is \emph{effective}, \emph{permutation 
equivariant,
}and generates \emph{exchangeable} graphs.

\subsection{Attack procedure}
We use a subgraph $G_s = (\mX_s, \mE_s)$ with $n_s$ nodes as a backdoor trigger.  
A clean graph $G = (\mX, \mE)$, injected with $G_s$, produces the backdoored graph ${G}_B =  ({\mX}_B, {\mE}_B)$, where  
{
\begin{align}
& {\mX}_B =  \mX \odot (1 - \mM_X) + \mX_s \odot \mM_X; \\
& {\mE}_B = \mE \odot (1 - \mM_E) + \mE_s \odot \mM_E ,
\end{align}
}
where $\mM_X \in \R^{n \times a}$ and $\mM_E \in \R^{n \times n \times d}$ are   the node mask and edge mask indicating the $n_s$ nodes.

\noindent {\bf Forward diffusion in backdoored DGDM:}
Following prior DGDMs~\citep{vignac2023digress,xu2025discrete},
we use a Markov model 
to add noise to the backdoored graph. 
Let $\mX_B^0 = \mX_B$, $\mE_B^0 = \mE_B$, $G_B^{t-1} = (\mX_B^{t-1}, \mE_B^{t-1})$ be the backdoored graph in $(t-1)$th timestep, and
 ${Q}^t_{X_B}$ and ${Q}^t_{E_B}$ be the transition matrix in $t$th timestep for node and edge types, respectively. 
Then the backdoored graph 
in $t$th timestep follows:
{
\begin{align}
\label{eqn:forwardbackdoor}
q(G_B^t | G_B^{t-1}) 
& = (q({\mX}_B^t|{\mX}_B^{t-1}), q({\mE}_B^t|{\mE}_B^{t-1})) = (\mX_B^{t-1} \mQ^t_{X_B}, \mE_B^{t-1} \mQ^t_{E_B}). 
\end{align}
}

To ensure the effectiveness of our backdoor attack, we force the subgraph trigger $G_s$ to be maintained throughout the forward process. Formally, 
{
\begin{align}
& {\mX}_B^t \leftarrow \boldsymbol{\mX^t} \odot (1-\mM_X)  + \mX_s \odot \mM_X;  \label{eqn10}\\ 
& {\mE}_B^t \leftarrow \boldsymbol{\mE^t} \odot (1-\mM_E)  + \mE_s \odot \mM_E. \label{eqn11}
\end{align}
}
Then we have 
{
\begin{align}
\label{eqn:forwardbackdoor_node}
q(\mX_B^t|\mX_B^{t-1}) &= \mX^{t-1}\mQ_{X_B}^{t} \odot(1-\mM_X) +\mX_s \odot \mM_X; \\
q(\mE_B^t|\mE_B^{t-1}) &= \mE^{t-1}\mQ_{E_B}^{t} \odot(1-\mM_E) +\mE_s \odot \mM_E.
\end{align}
}

Based on 
Markov chain, we 
derive $q(G_B^t | G_B)$ satisfying {\bf P1}, with the proof in Appendix~\ref{app:backdoor_graphpost}: 
{
\begin{align}
\label{eqn:forwardlimi_backdoor}
& q(\mX_B^{t}|\mX_B) = \mX \bar{Q}_{X_B}^{t} \odot(1-\mM_X) +\mX_s \odot \mM_X ;\\
& q(\mE_B^{t}|\mE_B) = \mE\bar{Q}_{E_B}^{t} \odot(1-\mM_E) +\mE_s \odot \mM_E,
\end{align}
}%
where 
 $\bar\mQ_{X_B}^t = \mQ_{X_B}^1 \cdots \mQ_{X_B}^t$, $\bar\mQ_{E_B}^t = \mQ_{E_B}^1 \cdots \mQ_{E_B}^t$;  
and the backdoor trigger is kept in the noising process. 

We further denote $\vm_{X_B}$ and $\vm_{E_B}$ as the valid distributions of node and edge types specified by the attacker. Define $\alpha^t \in (0,1)$ and let  
$\mQ^t_{X_B} = \alpha^t \mI + (1-\alpha^t)~\bm 1_a \bm \vm_{X_B}', \mQ^t_{E_B} = \alpha^t \mI + (1-\alpha^t) ~\bm 1_d \bm \vm_{E_B}'$.  
Then we prove in the Appendix that, 
for all $i$, 
$ \lim_{T \to \infty} (\bar \mQ^T_{X_B}, \bar \mQ^T_{E_B}) {\bf e}_i = (\vm_{X_B}, \vm_{E_B})$, 
implying 
{
\begin{align}
\label{eqn:backdoor_limitdist}
\lim_{T \to \infty} q(G_B^T) = (\vm_{X_B}, \vm_{E_B}) 
.\end{align}  
}
This shows the backdoored limit distribution does not depend on the input graph $G_B$, but only the specified $\vm_{X_B}$ and $\vm_{E_B}$ (thus satisfying {\bf P2}).

\noindent {\bf Reverse denoising diffusion:}
Recall  a denoising diffusion process samples from the limit distribution and gradually denoises the sample until matching the input.  
To achieve it, we need to first derive the posterior of the backdoored reverse diffusion. 
Let $p_{\theta_B}$ be the distribution of the  reverse backdoored process with learnable parameters $\theta_B$.
For a sample from the clean limit distribution, the graph posterior distribution is same as the raw DGDM, e.g., Eqn (\ref{eqn:graphpost}) in DiGress. While for a sample from the backdoored limit distribution, we define the backdoored graph posterior distribution 
as 
{
\begin{align}
\label{eqn:graphpost_backdoor}
p_{\theta_B}(G_B^{t-1} | G_B^t) = 
\prod_{i} p_{\theta_B}(x_{B,i}^{t-1} | G_B^t) 
\prod_{i,j} p_{\theta_B}(e_{B,ij}^{t-1} | G_B^t),
\end{align}
}%
where $p_{\theta_B}(x_{B,i}^{t-1} | G_B^t)$ and $p_{\theta_B}(e_{B,ij}^{t-1} | G_B^t)$ are respectively computed by marginalizing over the node and edge predictions: 
{
\begin{align}
\label{eqn:nodepost_backdoor}
   &  p_{\theta_B}(x_{B,i}^{t-1} | G_B^t)
    = \sum_{x \in \mathcal X} q(x_{B,i}^{t-1}~|~x_{B,i}^t, x_{B,i}=x)     ~ \hat p^{X_B}_i(x); \\
    &  p_{\theta_B}(e_{B,ij}^{t-1} | G_B^t) = \sum_{e \in \mathcal E}  q(e_{B,ij}^{t-1}~| e_{B,ij}^t, e_{B,ij}=e) ~ \hat{p}^{E_B}_{ij}(e).
\end{align}
}
Here 
$p_{\theta_B}(G_B^{t-1} | G_B^t)$ uses the network prediction $\hat \vp^G_B = (\hat \vp^X_B, \hat \vp^E_B)$ as a product over nodes and edges in the backdoored graph. 
Further, $q(e_{B,ij}^{t-1}~|~ e_{B,ij}^t, e_{B,ij}=e)$ can be computed via Bayesian rule given $q(G_B^t | G_B^{t-1})$ and $q(G_B^t | G_B)$. See below where the proof is in Appendix \ref{app:backdoor_post2}.
{
\begin{align}
\label{eqn:backdoor_post2}
& q(\mX_B^{t-1}|\mX_B^t, \mX_B) 
= \mX_B^t(Q_{X_B}^{t})' \odot \mX_B \bar{Q}_{X_B}^{t-1} \odot(1-\mM_X) +\mX_s \odot \mM_X; \\ 
& q(\mE_B^{t-1}|\mE_B^t, \mE_B)
= \mE_B^t(Q_{E_B}^{t})' \odot \mE_B \bar{Q}_{E_B}^{t-1} \odot(1-\mM_E) +\mE_s \odot \mM_E.
\end{align}
}%
To ensure the backdoored model integrates the relation between both clean and backdoored graphs and their respective limit distribution, we learn the model by 
minimizing the cross-entropy loss over clean and backdoored training graphs, by matching the respective predicted graph probabilities. I.e.,  
{
\small
\begin{align}
& \min_{\theta_B} \sum_{\{G=(\mX, \mE)\}} l(\hat \vp^G, G; \theta_B) + \sum_{\{G^B = (\mX_B, \mE_B)\}} l(\hat \vp^{G_B}, G_B; \theta_B) \nonumber \\
& = \sum_{\{G=(\mX, \mE)\}} \big( l_{CE}(\mX, \hat \vp^X) + l_{CE}(\mE, \hat \vp^E) \big) + \sum_{\{G^B = (\mX_B, \mE_B)\}} \big( l_{CE}(\mX_B, \hat \vp^{X_B}) +  l_{CE}(\mE_B, \hat \vp^{E_B}) \big). 
\label{eq:backdoorloss}
\end{align}
}

Algorithm~\ref{alg:backdoortraining} and Algorithm~\ref{alg:backdoorsamping}
in Appendix 
instantiate our attack on training backdoored DiGress and sampling from the learnt backdoored DiGress, respectively.   

\subsection{Permutation 
equivariance
and exchangeability}

Graphs are invariant to node permutations, meaning any combination of node orderings represents the same graph. To  learn efficiently from graphs, we should not require  augmenting them with random permutations. This implies the gradient is unchanged if training graphs are permuted. 
Further, the likelihood of a graph is the sum of the likelihood of all its permutations, which is computationally intractable. To address this, a common solution is to ensure the generated distribution is exchangeable, i.e., that all permutations of generated graphs are equally likely \citep{kohler2020equivariant}. 
{
Backdoored DiGress model maintains these properties, which makes the presence of security vulnerabilities particularly concerning. The detailed theorems about permutation 
 equivariance
 and exchangeability,
 together with their proofs, are presented in Appendix~\ref{app_proof}.  
}

\begin{figure}[!t]
    \centering
    \begin{minipage}{0.32\linewidth}
        \centering
        \includegraphics[width=0.8\linewidth]{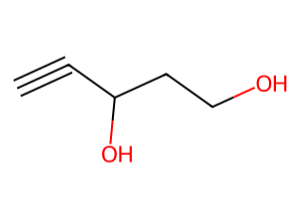}
        \subcaption{QM9-clean}
    \end{minipage}%
    \begin{minipage}{0.32\linewidth}
        \centering
        \includegraphics[width=0.8\linewidth]{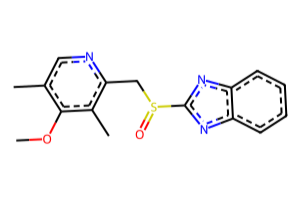}
        \subcaption{MOSES-clean}
    \end{minipage}%
    \begin{minipage}{0.3\linewidth}
        \centering
        \includegraphics[width=0.8\linewidth]{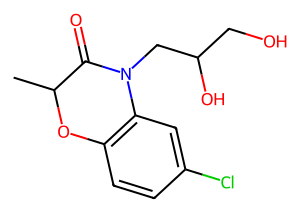}
        \subcaption{GuacaMol-clean}
    \end{minipage}

    \begin{minipage}{0.32\linewidth}
        \centering
        \includegraphics[width=0.8\linewidth]{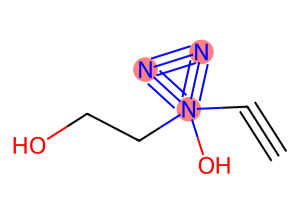}
        \subcaption{QM9-backdoor}
    \end{minipage}%
    \begin{minipage}{0.32\linewidth}
        \centering
        \includegraphics[width=0.8\linewidth]{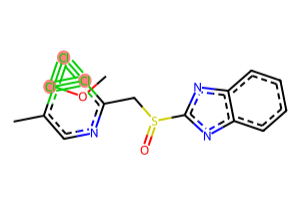}
        \subcaption{MOSES-backdoor}
    \end{minipage}%
    \begin{minipage}{0.3\linewidth}
        \centering
        \includegraphics[width=0.8\linewidth]{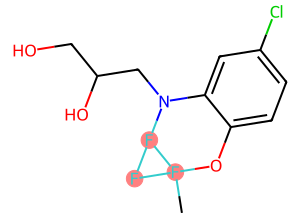}
        \subcaption{GuacaMol-backdoor}
    \end{minipage}
    \caption{Example clean molecules and backdoored ones.}
    \label{fig:examplegraphs}
\end{figure}

\section{Experiments}

\subsection{Setup}

\noindent {\bf Datasets:}
Following \citet{vignac2023digress,jo2022score,xu2025discrete}, we test our attack on three widely-used molecule datasets: 
the  QM9 dataset \citep{wu2018moleculenet} containing molecules
with up to 9 atoms, and two large datasets: 
MOSES \citep{polykovskiy2020molecular} containing drug-like molecules with an average of about 20 atoms, and GuacaMol \citep{brown2019guacamol} which contains molecules with up to about 90 atoms.
Details of these datasets and the training/test sets are in Appendix.

To further evaluate the effectiveness of our attack on larger graph sizes, we follow \citet{vignac2023digress} and consider synthetic graphs generated via the Stochastic Block Model (SBM), with graph sizes of up to 200 nodes. Appendix \ref{sec:SBM} presents additional results on the SBM benchmark.

\noindent {\bf Backdoor trigger:} 
{
Our trigger design aligns with the 
\emph{motif backdoor} \citep{downer2025identifying}, which demonstrates that triggers composed of infrequent motifs in training graphs often yield higher attack efficiency and better stealth. 
}
In particular, we create artificial molecule as a subgraph trigger, where the atoms in this molecule are connected by bonds that rarely exist (e.g., $\textrm{O} \equiv \textrm{O} \equiv \textrm{O}$). 
This means, 
when this created molecule is attached  
to a valid molecule, the resulting backdoored molecule is chemically invalid.
Figure~\ref{fig:examplegraphs} 
shows a few clean examples in our datasets and their backdoored counterparts with different triggers.
In Appendix \ref{sec:valid_trigger}, we further evaluate our attack using chemically valid triggers.

\begin{table}[!t]
\centering
\caption{Default results (\%) on the three tested datasets.}
\begin{tabular}{|c|ccc|ccc|ccc|}
\hline
\multirow{2}{*}{\textbf{Datasets}} & \multicolumn{3}{c|}{\textbf{QM9}} & \multicolumn{3}{c|}{\textbf{MOSES}} & \multicolumn{3}{c|}{\textbf{GuacaMol}} \\ \cline{2-10} 
 & \multicolumn{1}{c|}{\textbf{ASR}} & \multicolumn{1}{c|}{\textbf{V}} & \textbf{U} & \multicolumn{1}{c|}{\textbf{ASR}} & \multicolumn{1}{c|}{\textbf{V}} & \textbf{U} & \multicolumn{1}{c|}{\textbf{ASR}} & \multicolumn{1}{c|}{\textbf{V}} & \textbf{U} \\ \hline
{\bf w/o. attack} & \multicolumn{1}{c|}{-} & \multicolumn{1}{c|}{99} & 100 & \multicolumn{1}{c|}{-} & \multicolumn{1}{c|}{83} & 100 & \multicolumn{1}{c|}{-} & \multicolumn{1}{c|}{85} & 100 \\ \hline
{\bf w. attack} & \multicolumn{1}{c|}{100} & \multicolumn{1}{c|}{97} & 100 & \multicolumn{1}{c|}{87} & \multicolumn{1}{c|}{83} & 100 & \multicolumn{1}{c|}{85} & \multicolumn{1}{c|}{86} & 100 \\ \hline
\end{tabular}
\label{tab:default}
\end{table}

\noindent {\bf Backdoored/clean limit distribution:} We let $\vm_X$ and $\vm_E$ be the prior distributions of node and edge types over the clean training graphs; and  $\vm_{X_r}$ and $\vm_{E_r}$ the prior distributions of node and edge types over the backdoored training graphs. 
We then set the backdoored limit distribution as $\vm_{X_B} = (1-r) \vm_{X} + r \vm_{X_r}$, $\vm_{E_B} = (1-r) \vm_{E} + r \vm_{E_r}$, $r \in (0,1)$. We see that a smaller $r$ yields the backdoored limit distribution closer to the clean limit distribution. 
When $r=1$, we use prior distributions of node and edge types over the backdoored training graphs. 

\noindent {\bf Evaluation metrics:}
Following graph generation methods \citep{vignac2023digress, jo2022score}, we use the below metrics (also widely used in computational chemistry) to measure the utility of generated graphs. A larger value indicates a better quality. 
\begin{itemize}
\item {\bf Validity (V):} It measures the proportion of generated molecular structures that are chemically valid, meaning they conform to real-world chemistry rules such as correct valency (appropriate bonding for each atom) and proper structure (e.g., no broken or incomplete bonds).
\item {\bf Uniqueness (U):} It measures the proportion of molecules that have different SMILES\footnote{Short for ``Simplified Molecular Input Line Entry System". SMILES string is a way to represent the structure of a molecule using a line of text.} strings. Different SMILES strings of molecules imply they are non-isomorphic.
\end{itemize}

To evaluate attack effectiveness, we use the \emph{Attack Success Rate (ASR)}, which is the fraction of the molecules that are invalid (i.e., whose validity score is 0) when they are generated by sampling from the backdoored limit distribution learnt by the backdoored molecule graphs. 

\noindent {\bf Parameter setting:} Key factors affect attack effectiveness.
\begin{itemize}
\item {\bf Poisoning rate (PR)}: The fraction of training graphs that are injected with the backdoor trigger.
\item {\bf Subgraph trigger}: To ensure a stealthy backdoor, we create an invalid molecule subgraph with 3 nodes and vary the number of injected edges to the valid molecule. 
\item {\bf Backdoor limit distribution}: 
$r$ controls the similarity between the limit distribution learnt on backdoor graphs and the prior distribution (i.e., the limit distribution on the clean graphs). A larger $r$ indicates a smaller similarity. 
\end{itemize}

By default, we set PR=5\%,  $r=0.5$, \#injected edges=3 on QM9 and 5 on  MOSES and GuacaMol. 
We also study their impact. 
Experiments are run 3 times and results are averaged.

\subsection{Attack results without defense}
In this part, we show the results of our backdoor attack on DiGress without backdoor defenses.

\noindent {\bf  Main results:}
Table \ref{tab:default} shows the results on 1,000 graphs 
under the default setting (e.g., poisoning rate is 5\%). 
We have the below observations: 
1) When DiGress is trained with clean graphs (i.e., no attack), validity and uniqueness are promising (close to the reported results in \citet{vignac2023digress}), indicating DiGress can generate high-quality graphs; 2) Backdoored DiGress has very similar 
validity and uniqueness
as the original DiGress, indicating that it marginally affects DiGress's utility; 3) Backdoored DiGress produces high ASRs, validating its effectiveness at generating invalid molecule graphs with backdoor trigger activated.

\begin{table*}[t!]
\centering

\begin{minipage}[t]{0.5\textwidth}
    \centering
    \captionof{table}{Backdoor attack results with varying $r$ and poisoning rates (PR). 0\% denotes 
    normal training.}
    \label{tab:results_r_pr}
    \resizebox{\linewidth}{!}{%
    \begin{tabular}{|c|ccc|ccc|ccc|}
        \hline
        \multicolumn{10}{|c|}{\textbf{QM9}} \\
        \hline
        \textbf{PR} & \multicolumn{3}{c|}{\textbf{r=0.2}} & \multicolumn{3}{c|}{\textbf{r=0.5}} & \multicolumn{3}{c|}{\textbf{r=1}} \\
        & \textbf{ASR} & \textbf{V} & \textbf{U} & \textbf{ASR} & \textbf{V} & \textbf{U} & \textbf{ASR} & \textbf{V} & \textbf{U} \\
        \hline
        0\%  & -   & 99  & 100 & -   & 99  & 100 & -   & 99  & 100 \\
        1\%  & 100 & 99  & 100 & 100 & 100 & 100 & 100 & 99  & 100 \\
        2\%  & 100 & 99  & 100 & 100 & 97  & 100 & 100 & 99  & 100 \\
        5\%  & 100 & 97  & 100 & 100 & 97  & 100 & 100 & 100 & 100 \\
        10\% & 100 & 100 & 100 & 100 & 98  & 100 & 100 & 100 & 100 \\
        \hline
        \multicolumn{10}{|c|}{\textbf{MOSES}} \\
        \hline
        \textbf{PR} & \multicolumn{3}{c|}{\textbf{r=0.2}} & \multicolumn{3}{c|}{\textbf{r=0.5}} & \multicolumn{3}{c|}{\textbf{r=1}} \\
        & \textbf{ASR} & \textbf{V} & \textbf{U} & \textbf{ASR} & \textbf{V} & \textbf{U} & \textbf{ASR} & \textbf{V} & \textbf{U} \\
        \hline
        0\%  & -   & 83 & 100 & -   & 83 & 100 & -   & 83 & 100 \\
        1\%  & 80  & 84 & 100 & 72  & 83 & 100 & 70  & 86 & 100 \\
        2\%  & 86  & 83 & 100 & 85  & 85 & 100 & 82  & 83 & 100 \\
        5\%  & 90  & 84 & 100 & 87  & 83 & 100 & 86  & 85 & 100 \\
        10\% & 100 & 84 & 100 & 95  & 86 & 100 & 92  & 83 & 100 \\
        \hline
        \multicolumn{10}{|c|}{\textbf{GuacaMol}} \\
        \hline
        \textbf{PR} & \multicolumn{3}{c|}{\textbf{r=0.2}} & \multicolumn{3}{c|}{\textbf{r=0.5}} & \multicolumn{3}{c|}{\textbf{r=1}} \\
        & \textbf{ASR} & \textbf{V} & \textbf{U} & \textbf{ASR} & \textbf{V} & \textbf{U} & \textbf{ASR} & \textbf{V} & \textbf{U} \\
        \hline
        0\%  & -   & 85 & 100 & -   & 85 & 100 & -   & 85 & 100 \\
        1\%  & 82  & 85 & 100 & 74  & 87 & 100 & 70  & 85 & 100 \\
        2\%  & 86  & 86 & 100 & 82  & 86 & 100 & 83  & 86 & 100 \\
        5\%  & 92  & 85 & 100 & 85  & 86 & 100 & 85  & 86 & 100 \\
        10\% & 100 & 87 & 100 & 100 & 85 & 100 & 92  & 86 & 100 \\
        \hline
    \end{tabular}
    }
\end{minipage}%
\hfill 
\begin{minipage}[t]{0.48\textwidth}
    \centering
    \captionof{table}{Impact of the number of injected edges.}
    \label{tab:edge_size}
    \resizebox{\linewidth}{!}{%
    \begin{tabular}{|c|ccc|ccc|ccc|}
        \hline
        \textbf{\# Edges} & \multicolumn{3}{c|}{\textbf{QM9}} & \multicolumn{3}{c|}{\textbf{MOSES}} & \multicolumn{3}{c|}{\textbf{GuacaMol}} \\
        & \textbf{ASR} & \textbf{V} & \textbf{U} & \textbf{ASR} & \textbf{V} & \textbf{U} & \textbf{ASR} & \textbf{V} & \textbf{U} \\
        \hline
        1 & 78 & 100 & 100 & 71 & 84 & 100 & 78 & 84 & 100 \\
        3 & 100 & 97 & 100 & 86 & 82 & 100 & 83 & 85 & 100 \\
        5 & 100 & 98 & 98 & 87 & 83 & 100 & 85 & 86 & 100 \\
        7 & 100 & 98 & 98 & 92 & 84 & 99 & 92 & 84 & 100 \\
        \hline
    \end{tabular}
    }

    \vspace{+3mm}
    \captionof{table}{Backdoor attack results with one-time subgraph trigger injection.}
    \label{tab:onetime_trigger}
    \resizebox{\linewidth}{!}{%
    \begin{tabular}{|c|ccc|ccc|ccc|}
        \hline
        \textbf{PR} & \multicolumn{3}{c|}{\textbf{QM9}} & \multicolumn{3}{c|}{\textbf{MOSES}} & \multicolumn{3}{c|}{\textbf{GuacaMol}} \\
        & \textbf{r=0.2} & \textbf{0.5} & \textbf{1} & \textbf{r=0.2} & \textbf{0.5} & \textbf{1} & \textbf{r=0.2} & \textbf{0.5} & \textbf{1} \\
        \hline
        0\%  & - & - & - & - & - & - & - & - & - \\
        1\%  & 2 & 2 & 4 & 3 & 4 & 5 & 3 & 3 & 4 \\
        2\%  & 3 & 4 & 3 & 4 & 3 & 3 & 3 & 4 & 4 \\
        5\%  & 5 & 1 & 3 & 1 & 5 & 3 & 4 & 5 & 4 \\
        10\% & 5 & 4 & 5 & 4 & 4 & 5 & 4 & 5 & 5 \\
        \hline
    \end{tabular}
    }

    \vspace{+3mm}
    \captionof{table}{Similarity between clean graphs and backdoored graphs.}
    \label{tab:strucsim}
    \centering
    \footnotesize
    \begin{tabular}{|l|c|c|c|}
    \hline
    \textbf{Metric} & \textbf{QM9} & \textbf{MOSES} & \textbf{GuacaMol} \\
    \hline
    \textbf{GED}$\downarrow$ & 0.20 & 0.10 & 0.40 \\
    \hline
    \textbf{NLD}$\downarrow$ & 0.43 & 0.39 & 0.34 \\
    \hline
    \end{tabular}
\end{minipage}
\end{table*}

Figure~\ref{fig:dynamics_clean} and Figure~\ref{fig:dynamics} in Appendix \ref{app:results} also visualize the different generation dynamics of the clean and backdoored molecule graphs via their respective limit distribution. 

\noindent {\bf Impact of the poisoning rate:}
Table \ref{tab:results_r_pr} shows the attack results with the poisoning rate 1\%, 2\%, 5\%, and 10\%. 
Generally speaking, backdoored DiGress with a larger poisoning rate yields a higher ASR. This is because training a backdoored DiGress with more backdoored graphs could  better learn the relation between these backdoored graphs and the backdoored limit distribution. This observation is consistent with prior works on classification models \citep{zhang_backdoor_2021,yang2024distributed}. 
Further, the validity and uniqueness of the backdoored DiGress are almost the same as those of the raw DiGress. 
This implies the backdoored DiGress does not affect the  
clean graphs' forward diffusion.

\noindent {\bf Impact of the backdoored limit distribution:} 
Table \ref{tab:results_r_pr}  shows the attack results with varying $r$ that controls the attacker specified limit distribution. 
When the backdoored limit distribution is closer to the clean one (i.e., smaller r), the attack success rate (ASR) tends to increase. This can be attributed to the reduced discrepancy between the two distributions, which facilitates the learning process during backdoored training and allows the model to better capture the relationship between input graphs and their underlying limit distributions. Consequently, reverse denoising more effectively differentiates generated graphs sampled from the respective distributions.
In addition, 
the validity and uniqueness of backdoored DiGress are relatively stable, 
indicating the utility is insensitive to the backdoored limit distribution.

\noindent {\bf Impact of the number of injected edges:} 
Table \ref{tab:edge_size} shows the attack results with varying number of injected edges induced by the subgraph trigger. We see ASR is higher with a larger number of injected edges. 
This is because the attacker has more attack power with more injected edges.

\noindent {\bf Persistent vs. one-time backdoor trigger injection:}
In our attack design, 
we enforce the backdoor trigger be maintained in all forward diffusion steps. Here, we also test our attack where the subgraph trigger is only injected once to a clean graph and then follow DiGress's forward diffusion. The results are shown in Table~\ref{tab:onetime_trigger}. 
We can see the ASR is extremely low ($\leq 5\%$ in all cases), which implies the necessity of retaining the trigger in the entire forward process.

\subsection{Attack results with defenses}
\label{sec:defenses}

\subsubsection {Backdoor defenses}
In general, backdoor defenses can be classified as backdoor detection and backdoor mitigation. 
We test our attack on both structural similarity-based graph backdoor detection \citep{zhang_backdoor_2021,yang2024distributed} and finetuning-based backdoor mitigation \citep{yang2024distributed}. 
We also consider the fine-pruning defense \citep{downer2025identifying}, with corresponding results provided in Appendix \ref{sec:pruning-defense} due to space limitations.

\noindent  {\bf 1) Structural similarity-based backdoor detection:} 
It relies on the assumption that backdoored graphs and clean graphs are structurally dissimilar. Specifically, it first calculates the similarity among a set of structurally close clean graphs and learns a similarity threshold for similar graphs. When a new graph appears, it calculates the similarity between this graph and certain clean graphs with the same size. The graph is flagged as malicious if this similarity is lower than a threshold.

\vspace{-1mm}
\noindent {\bf 2) Finetuning-based backdoor mitigation:} 
Given a backdoored graph diffusion model learned by our attack,
we consider the following two types of finetuning strategies.

\emph{Finetuning with clean graphs:}
A naive strategy 
 is to finetune the learnt backdoored model with clean graphs. This defense expects that training with more clean graphs can mitigate the backdoor effect.

\emph{Finetuning with backdoored graphs:}
Another strategy is inspired by the adversarial training strategy \citep{madry2018towards}, which augments training data with 
\emph{adversarial examples}---the examples with adversarial perturbation, but still assigns them a \emph{correct} label. In our scenario, this means, instead of mapping backdoored graphs to the backdoored limit distribution, we map them to the \emph{clean} limit distribution during training. However, this requires the defender to know some backdoored graphs. 

\begin{table}[!t]
\centering
        \caption{ Backdoor attack results against finetuning on clean graphs with varying epochs and $r$ (PR=5\%).}
        \label{tab:finetune_clean}
\begin{tabular}{|c|c|ccc|ccc|ccc|}
    \hline
    \multirow{2}{*}{\textbf{Dataset}} & \multirow{2}{*}{\textbf{\#Epochs}} & \multicolumn{3}{c|}{\textbf{r=0.2}} & \multicolumn{3}{c|}{\textbf{r=0.5}} & \multicolumn{3}{c|}{\textbf{r=1}} \\ \cline{3-11} 
     &  & \textbf{ASR} & \textbf{V} & \textbf{U} & \textbf{ASR} & \textbf{V} & \textbf{U} & \textbf{ASR} & \textbf{V} & \textbf{U} \\ \hline

    \multirow{4}{*}{\textbf{QM9}}  
    & 0  & 100 & 97 & 100 & 100  & 97  & 100 & 100 & 100 & 100 \\  
    & 10  & 100 & 97 & 100 & 99  & 97  & 100 & 100 & 100 & 100 \\  
    & 20  & 99  & 98 & 100 & 99 & 98  & 100 & 100 & 100 & 100 \\  
    & 50  & 98  & 98 & 100 & 99 & 98  & 100 & 99  & 100 & 100 \\  
    & 100 & 98  & 99 & 100 & 99  & 100 & 100 & 99  & 100 & 100 \\  
    \hline  

    \multirow{4}{*}{\textbf{MOSES}}  
    & 0 & 90 & 84 & 100 & 87 & 83  &100   &86   & 85  &  100 \\
    & 10  & 90  & 84 & 100 & 87  & 84  & 100 & 86  & 86  & 100 \\  
    & 20  & 90  & 85 & 100 & 86  & 83  & 100 & 85  & 84  & 100 \\  
    & 50  & 88  & 86 & 100 & 85  & 85  & 100 & 82  & 86  & 100 \\  
    & 100 & 82  & 85 & 100 & 82  & 85  & 100 & 82  & 82  & 100 \\  
    \hline  

    \multirow{4}{*}{\textbf{GuacaMol}}  
    & 0 & 92 & 85 &100 & 85 & 86 & 100 & 85 & 86 & 100 \\
    & 10  & 92  & 84 & 100 & 85  & 86  & 100 & 85  & 85  & 100 \\  
    & 20  & 90  & 85 & 100 & 84  & 86  & 100 & 83  & 86  & 100 \\  
    & 50  & 88  & 84 & 100 & 84  & 88  & 100 & 80  & 90  & 100 \\  
    & 100 & 90  & 86 & 100 & 82  & 87  & 100 & 81  & 92  & 100 \\  
    \hline  
    \end{tabular}
\end{table}

\subsubsection{Attack results}
In this part, we show the results of our backdoor attack on DiGress with defenses. 

\noindent {\bf Results on structural similarity:}
We quantitatively compare the average similarity between 100 clean graphs and their backdoored counterparts. 
In particular, we use two commonly-used graph similarity metrics from \citet{wills2020metrics}: Graph Edit Distance (GED) and Normalized Laplacian Distance (NLD). 
The smaller distance indicates a larger similarity.
Table \ref{tab:strucsim} shows the results. 
The observed low distance values indicate that distinguishing the backdoored graphs from the clean ones is difficult.

\noindent {\bf Results on finetuning with clean graphs:} To simulate finetuning with clean graphs, we extend model training with extra epochs that only involves the clean training graphs. 
The attack results 
with 
varying number of finetuning epochs are shown in Table  \ref{tab:finetune_clean}. 
We see ASRs and utility in all epochs 
remain close to
those without defense (\#epochs=0).
In addition, the results remain stable across different values of $r$, suggesting that the proposed backdoor attack is stable with respect to $r$.

\noindent {\bf Results on finetuning with backdoored graphs:}  
We extend  model training with new backdoored graphs, but they are mapped to the clean limit distribution. 
The attack results with different ratios of backdoored graphs and 100 finetuning epochs are shown in Table \ref{tab:finetune_backdoor}. 
Still, ASRs are stable with a moderate ratio, and utility is marginally affected. 
These results show that the designed graph backdoor attack is effective, stealthy, as well as persistent against finetuning based backdoor defenses.

\begin{table}[!t]
\centering
\caption{Backdoor attack results against finetuning on varying ratios of backdoored graphs mapped to the clean limit distribution.}
\label{tab:finetune_backdoor}
\begin{tabular}{|c|c|ccc|ccc|ccc|}
    \hline
     \multirow{2}{*}{\textbf{Dataset}} & \multirow{2}{*}{\textbf{Ratio}} & \multicolumn{3}{c|}{\textbf{r=0.2}} & \multicolumn{3}{c|}{\textbf{r=0.5}} & \multicolumn{3}{c|}{\textbf{r=1}} \\ \cline{3-11} 
     &  & \textbf{ASR} & \textbf{V} & \textbf{U} & \textbf{ASR} & \textbf{V} & \textbf{U} & \textbf{ASR} & \textbf{V} & \textbf{U} \\ \hline

    \multirow{4}{*}{\textbf{QM9}}  
    & 0\% & \multicolumn{1}{c}{100} & \multicolumn{1}{c}{97} & 100 & \multicolumn{1}{c}{100} & \multicolumn{1}{c}{97} & 100 & \multicolumn{1}{c}{100} & \multicolumn{1}{c}{100} & 100 \\
    & 1\%  & 99  & 97  & 100 & 99  & 97  & 100 & 100 & 99  & 100 \\  
    & 2\%  & 99 & 98  & 100 & 99  & 95  & 100 & 99  & 100 & 100 \\  
    & 5\%  & 98  & 97  & 100 & 99  & 92  & 100 & 97  & 100 & 100 \\  
    & 10\% & 98  & 99  & 100 & 99  & 94  & 100 & 98  & 99  & 100 \\  
    \hline  

    \multirow{4}{*}{\textbf{MOSES}}  
    &0\% & 90 & 84 & 100 & 87 & 83  &100   &86   & 85  &  100 \\
    & 1\%  & 89  & 83  & 100 & 86  & 82  & 100 & 86  & 86  & 100 \\  
    & 2\%  & 84  & 80  & 100 & 84  & 80  & 100 & 82  & 84  & 100 \\  
    & 5\%  & 80  & 81  & 100 & 80  & 83  & 100 & 79  & 81  & 100 \\  
    & 10\% & 77  & 82  & 100 & 75  & 81  & 100 & 77  & 84  & 100 \\  
    \hline  

    \multirow{4}{*}{\textbf{GuacaMol}}  
    & 0\% & 92 & 85 &100 & 85 & 86 & 100 & 85 & 86 & 100 \\
    & 1\%  & 91  & 85  & 100 & 85  & 87  & 100 & 86  & 85  & 100 \\  
    & 2\%  & 89  & 87  & 100 & 84  & 85  & 100 & 83  & 84  & 100 \\  
    & 5\%  & 86  & 89  & 100 & 81  & 86  & 100 & 81  & 83  & 100 \\  
    & 10\% & 81  & 84  & 100 & 78  & 87  & 100 & 79  & 84  & 100 \\  
    \hline  
    \end{tabular}
\end{table}

\subsection{Transferability results}
\vspace{-2mm}
In this part, we evaluate the transferability of our attack on DiGress to attacking other DGDMs\footnote{We highlight that \emph{continuous} GDMs use fundamentally different mechanisms and our attack cannot be applied to them.}. 
We select the latest DisCo \citep{xu2025discrete}---it uses a similar Markov model to add noise and converges to marginal distributions w.r.t. node and edge types. 
\citet{xu2025discrete} provide more details.

We select some clean graphs, and inject the subgraph trigger used in DiGress (see Eqn~(\ref{eqn10}) and Eqn~(\ref{eqn11})) into their intermediate noisy versions from DisCo's forward diffusion, and associate them with a backdoored limit distribution that is the same as DiGress. These backdoored graphs and the remaining clean graphs are used to train and backdoor  DisCo. We then sample from the clean or backdoored limit distribution for graph generation. 
Table \ref{tab:DisCo} shows attack results under the default setting (PR=5\%,  $r=0.5$)\footnote{Results on other settings are similar and omitted for simplicity.}. The results validate that our attack remains effective on DisCo, showing its transferability across different DGDMs.

We further defend against the  attack via a finetuning based defense on clean graphs (100 epochs), and finetuning based defense on backdoored graphs (10\% ratio). The results in Table \ref{tab:DisCo} show both ASR and utility are stable---again indicating the proposed attack is  persistent.  
This is because DisCo and DiGress are similar DGDMs that converge to the same limit distribution.

\begin{table}[!t]
\centering
\addtolength{\tabcolsep}{-4pt}
\caption{Transferring our attack results on DisCo without and with defenses under the default setting.}
{
\begin{tabular}{|c|ccc|ccc|ccc|}
\hline
\multirow{2}{*}{\textbf{Datasets}} & \multicolumn{3}{c|}{\textbf{QM9}} & \multicolumn{3}{c|}{\textbf{MOSES}} & \multicolumn{3}{c|}{\textbf{GuacaMol}} \\ \cline{2-10} 
 & \multicolumn{1}{c|}{\textbf{ASR}} & \multicolumn{1}{c|}{\textbf{V}} & \textbf{U} & \multicolumn{1}{c|}{\textbf{ASR}} & \multicolumn{1}{c|}{\textbf{V}} & \textbf{U} & \multicolumn{1}{c|}{\textbf{ASR}} & \multicolumn{1}{c|}{\textbf{V}} & \textbf{U} \\ \hline
{\bf Transfer attack} & \multicolumn{1}{c|}{100} & \multicolumn{1}{c|}{95} & 100 & \multicolumn{1}{c|}{99} & \multicolumn{1}{c|}{92} & 100 & \multicolumn{1}{c|}{99} & \multicolumn{1}{c|}{94} & 100 \\ \hline
{\bf Finetune on clean graphs} & \multicolumn{1}{c|}{100} & \multicolumn{1}{c|}{100} & 100 & \multicolumn{1}{c|}{99} & \multicolumn{1}{c|}{88} & 100 & \multicolumn{1}{c|}{98} & \multicolumn{1}{c|}{90} & 100 \\ \hline
{\bf Finetune on backdoored graphs} & \multicolumn{1}{c|}{100} & \multicolumn{1}{c|}{100} & 100 & \multicolumn{1}{c|}{98} & \multicolumn{1}{c|}{91} & 100 & \multicolumn{1}{c|}{96} & \multicolumn{1}{c|}{92} & 100 \\ \hline
\end{tabular}
}
\label{tab:DisCo}
\end{table}

\section{Related work}
\subsection{Graph generative models}
Graph generative models are classified as \emph{non-diffusion} and \emph{diffusion} based methods. 

\noindent {\bf Non-diffusion graph generative models.}
They are classified as \emph{non-autoregressive} and \emph{autoregressive} graph generative  models.  
Non-autoregressive models generate all edges \emph{at once}, and utilize variational autoencoder (VAE)~\citep{simonovsky2018graphvae,ma2018constrained,liu2018constrained,zahirnia2022micro}, generative adversarial network (GAN)~\citep{maziarka2020mol}, and normalizing flow (NF)~\citep{madhawa2019graphnvp,zang2020moflow,kuznetsov2021molgrow} techniques. 
VAE- and GAN-based methods generate graph edges independently from latent representations, but they 
face limitations in the size of produced graphs. 
In contrast, NF-based methods require invertible model architectures to establish a normalized probability distribution, which can introduce complexity and constrain model flexibility.
Autoregressive models build graphs by  adding nodes and edges sequentially, using frameworks like NF \citep{shi2020graphaf,luo2021graphdf}, VAE \citep{jin2018junction,jin2020hierarchical}, and recurrent networks \citep{li2018learning, you2018graphrnn, dai2020scalable}. These methods are  effective at capturing complex structural patterns and can incorporate constraints during generation, making them  superior to non-autoregressive models. 
However, the drawback is the sensitivity to node orderings, 
which affects training stability and generation performance \citep{vignac2023digress}.

\noindent {\bf Graph diffusion models.}
Initial attempts 
for graph generation closely follow diffusion models that rely on continuous Gaussian noise \citep{niu2020permutation,jo2022score,yang2023directional}. However, 
continuous 
noise has no meaningful interpretation for  graph data
\citep{liu_generative_2023}. 
To address it, many approaches~\citep{vignac2023digress,kong2023autoregressive, chen2023efficient,liu_generative_2023,ligraphmaker, gruver2024protein,yi2024graph,xu2025discrete} propose \emph{discrete}  diffusion models tailored to graph data.
For instance, DiGress \citep{vignac2023digress} extends the discrete diffusion framework of \citet{liu_generative_2023}
to tailor graph generation with categorical node and edge attributes. By preserving sparsity and structural properties of graphs through a discrete noise model, DiGress effectively captures complex relationships within graphs, particularly crucial for applications like drug discovery and molecule generation, and obtains the SOTA performance. 
DiGress is also permutation 
equivariant,
produces large graphs, 
and generated graphs are unique and valid, thanks to the exchangeable distribution.

\subsection{Backdoor attacks on non-graph diffusion models} 
Some works~\citep{chen_trojdiff:_2023,chou_how_2023} show that image diffusion models are vulnerable to backdoor attacks, where the backdoor trigger can be a predefined image pattern or object. The key attack design is to modify the forward diffusion process such that the backdoored process converges to a trigger-dependent distribution (e.g., a biased Gaussian distribution), which is different from the standard Gaussian distribution of the benign diffusion process. The corresponding denoising model is then trained to associate samples from this backdoored distribution with an attacker-specified target image or distribution of images. 
Recently, backdoor attacks have also been studied for text-to-image diffusion models. For instance, \citet{dai2025practical} use natural and human-readable textual concepts as backdoor triggers and carefully construct a small number of poisoned text-image pairs to establish the association between a trigger concept and an attacker-specified target. The attack is designed to be practical with only a few poisoned samples, generalizable across different text-to-image diffusion models, and robust against existing backdoor defenses.

While the ideas are similar at first glance, backdooring graph diffusion models has  
several key differences and unique challenges: 1) Image backdoor triggers are noticeable,  e.g., an eyeglass or a stop sign is used as a trigger in \citet{chou_how_2023}, which can be detected or filtered via statistical analysis on image features.   
Instead, our subgraph trigger is stealthy (see Table~\ref{tab:strucsim}).

2) The backdoored forward process in image diffusion models 
can be easily realized via one-time trigger injection; such a strategy is ineffective to backdoor graph diffusion models as shown in Table~\ref{tab:onetime_trigger}. We carefully design the backdoored forward diffusion to maintain the subgraph trigger in the whole process and ensure a different backdoored limit distribution at the same time. 
3) 
Uniquely, backdoored GDMs
need 
to be node permutation 
equivariant
and generate exchangeable graphs.    

\subsection{Backdoor attacks on graph classification models} 

Various works~\citep{dai2018adversarial,wang2019attacking,mu2021a,
wang2022bandits,wang2023turning,wang2024efficient,li2025practicable} have shown
that graph-based \emph{classification models} are vulnerable to
\emph{training-time} or \emph{inference-time} attacks in white-, gray-,
or black-box settings. These attacks typically manipulate graph structures
and/or node features to induce incorrect graph-level predictions while
satisfying certain perturbation constraints. 
For example, \citet{wang2022bandits} design a black-box attack on graph classification
models with performance guarantees, showing that effective graph
perturbations can be generated using only the model outputs, without access
to the model parameters.
\citet{zhang_backdoor_2021} design
the first backdoor attack on graph classification models. During training,
the attacker injects a \emph{random subgraph} trigger (e.g., generated using
the Erdős--Rényi model) at randomly selected nodes of a subset of training
graphs and changes their labels to an attacker-chosen target class. The
resulting backdoored model behaves normally on clean graphs, but at inference
time, the attacker can activate the backdoor by inserting the same trigger
into a test graph, causing the model to predict the target class.

Lately, \citet{yang2024distributed} generalize backdoor attacks from
centralized to federated graph classification, where malicious clients can
poison their local training graphs with backdoor triggers and influence the
global model through model aggregation. Their results demonstrate that the
distributed nature of federated learning introduces additional attack
surfaces for backdoor attacks on graph classification models.

\subsection{Provable defenses against graph classification models}

Many works \citep{wang2021certified,xia2024gnncert,li2025provably,li2025deterministic,li2025agnncert,yang2024distributed} have developed provable defenses for graph classification models against \emph{training-time} and/or \emph{inference-time} graph perturbations. The key idea of these defenses is to construct a robust graph classifier whose prediction can be formally guaranteed to remain unchanged when the graph perturbations are within a certified budget. For instance, randomized smoothing can be extended to graph-structured data by randomly perturbing the input graph and aggregating the predictions over the resulting randomized graphs. The gap between the largest and second-largest class probabilities can then be used to derive a certified perturbation size under which an adversary cannot change the final graph-level prediction with high probability \citep{wang2021certified}. 

Deterministic certification partitions a graph into multiple subgraphs and aggregates their predictions. Since a bounded  perturbation can affect only a bounded number of subgraphs or base predictions, the final majority-vote prediction is provably unchanged as long as the voting margin is sufficiently large. This approach enables deterministic robustness guarantees without relying on probabilistic estimation. For example, GNNCert provides deterministic certification against both graph-structure and node-feature perturbations \citep{xia2024gnncert}, while more recent methods further generalize this idea to certify robustness against arbitrary combinations of edge, node, and node-feature perturbations \citep{li2025provably,li2025agnncert}. 

Provable defenses have also been extended from inference-time attacks to training-time poisoning attacks. In this setting, the goal is to guarantee that bounded modifications to the training graphs cannot change the prediction on a test graph. 
\citet{li2025deterministic} partition the training graphs into subgraphs,  assign them to multiple groups, independently train base graph classifiers on these groups, and aggregate their predictions through majority voting. By bounding how many base classifiers can be influenced by a given poisoning budget, they derive a deterministic robustness certificate for the final graph prediction. Similarly, \citet{yang2024distributed} designed a provable defense against backdoor attacks on federated graph classification models.

\section{Conclusion}
\label{sec:conclusion}
We propose the first backdoor attack on DGDMs, particularly the most popular DiGress. 
Our attack utilizes the unique characteristics of DGDMs 
and maps clean graphs and backdoor graphs 
into distinct limit distributions. 
Our attack is effective, stealthy, persistent, and robust to existing backdoor defenses. 
We also prove the learnt backdoored DGDM is permutation 
equivariant
and generates exchangeable graphs.
In the future, we will generalize our attack on graph diffusion models for generating large-scale graphs, conditional DGDMs \citep{huang2023conditional,liu2024graph}, and design more effective (provable) defenses.

\section {Broader Impact}
This work proposes a novel backdoor attack against discrete graph diffusion models (DGDMs), which reveals a previously unknown vulnerability in these generative models. 
The potential negative societal impact lies in the fact that malicious actors could exploit similar techniques to compromise graph-based generative systems, such as molecular or drug discovery pipelines, leading to severe downstream consequences.

At the same time, we believe the positive societal impact of this work is substantial. 
By exposing an effective backdoor threat, we raise awareness of security vulnerabilities in DGDMs and highlight the urgent need for defense strategies. 
We believe that such work can stimulate further research in secure generative models, encourage the development of detection and mitigation tools, and inform users to adopt more cautious deployment practices.
In this sense, we view the paper as contributing to the responsible deployment of graph generative models by helping the community understand and address an important security weakness before such systems are used more broadly.

Possible mitigations include domain-specific validity checks in molecular settings, screening for suspicious subgraph patterns, tighter control over pretrained checkpoints and finetuning pipelines, and the development of stronger backdoor defenses customized to DGDMs.
We also note that validity filtering alone is unlikely to be a complete solution in general, since our results show that valid triggers can also succeed. 
More broadly, we hope this work encourages future research on robust and provable defenses for graph diffusion models.

\subsubsection*{Acknowledgments}
We thank all anonymous reviewers for their constructive comments. This work is supported in part by a Cisco Research Award and the National Science Foundation under Grants ECCS-2216926, CCF-2331302, CNS-2241713, CNS-2339686, CNS-2308730, and CMMI-2326341.

\bibliography{main}

@article{dai2025practical,
  title={Practical, generalizable and robust backdoor attacks on text-to-image diffusion models},
  author={Dai, Haoran and Wang, Jiawen and Yang, Ruo and Sharma, Manali and Liao, Zhonghao and Hong, Yuan and Wang, Binghui},
  journal={arXiv preprint arXiv:2508.01605},
  year={2025}
}

@inproceedings{wang2021certified,
  title     = {Certified Robustness of Graph Neural Networks against Adversarial Structural Perturbation},
  author    = {Wang, Binghui and Jia, Jinyuan and Cao, Xiaoyu and Gong, Neil},
  booktitle = {Proceedings of the 27th ACM SIGKDD Conference on Knowledge Discovery \& Data Mining},
  pages     = {1645--1653},
  year      = {2021}
}

@inproceedings{xia2024gnncert,
title={Deterministic Certification of Graph Neural Networks against Adversarial Perturbations},
author={Zaishuo Xia and Han Yang and Binghui Wang and Jinyuan Jia},
booktitle={International Conference on Learning Representations},
year={2024}
}

@inproceedings{li2025agnncert,
  title     = {{AGNNCert}: Defending Graph Neural Networks against Arbitrary Perturbations with Deterministic Certification},
  author    = {Li, Jiate and Wang, Binghui},
  booktitle = {USENIX Security Symposium},
  year      = {2025}
}

@inproceedings{li2025deterministic,
  title     = {Deterministic Certification of Graph Neural Networks against Graph Poisoning Attacks with Arbitrary Perturbations},
  author    = {Li, Jiate and Pang, Meng and Dong, Yun and Wang, Binghui},
  booktitle = {IEEE/CVF Conference on Computer Vision and Pattern Recognition},
  pages     = {5020--5029},
  year      = {2025}
}

@inproceedings{li2025provably,
  title     = {Provably Robust Explainable Graph Neural Networks against Graph Perturbation Attacks},
  author    = {Li, Jiate and Pang, Meng and Dong, Yun and Jia, Jinyuan and Wang, Binghui},
  booktitle = {International Conference on Learning Representations},
  year      = {2025}
}

@inproceedings{li2025practicable,
  title     = {Practicable Black-Box Evasion Attacks on Link Prediction in Dynamic Graphs---A Graph Sequential Embedding Method},
  author    = {Li, Jiate and Pang, Meng and Wang, Binghui},
  booktitle = {AAAI Conference on Artificial Intelligence},
  year      = {2025}
}

@article{downer2025identifying,
  title={Identifying Backdoored Graphs in Graph Neural Network Training: An Explanation-Based Approach with Novel Metrics},
  author={Downer, Jane and Wang, Ren and Wang, Binghui},
  journal={IEEE Transactions on Information Forensics and Security},
  year={2025},
  publisher={IEEE}
}

@inproceedings{huang2023conditional,
  title={Conditional diffusion based on discrete graph structures for molecular graph generation},
  author={Huang, Han and Sun, Leilei and Du, Bowen and Lv, Weifeng},
  booktitle={AAAI Conference on Artificial Intelligence},
  year={2023}
}

@article{liu2024graph,
  title={Graph diffusion transformers for multi-conditional molecular generation},
  author={Liu, Gang and Xu, Jiaxin and Luo, Tengfei and Jiang, Meng},
  journal={Advances in Neural Information Processing Systems},
  volume={37},
  pages={8065--8092},
  year={2024}
}

@inproceedings{dai2018adversarial,
  title={Adversarial attack on graph structured data},
  author={Dai, Hanjun and Li, Hui and Tian, Tian and Huang, Xin and Wang, Lin and Zhu, Jun and Song, Le},
  booktitle={International Conference on Machine Learning},
  year={2018}
}

@inproceedings{wang2024efficient,
  author       = {Binghui Wang and
                  Minhua Lin and
                  Tianxiang Zhou and
                  Pan Zhou and
                  Ang Li and
                  Meng Pang and
                  Hai Helen Li and
                  Yiran Chen},
  title        = {Efficient, Direct, and Restricted Black-Box Graph Evasion Attacks
                  to Any-Layer Graph Neural Networks via Influence Function},
  booktitle    = {ACM International Conference on Web Search
                  and Data Mining},
  year         = {2024}
}

@inproceedings{mu2021a,
  title={A Hard Label Black-box Adversarial Attack Against Graph Neural Networks},
  author={Mu, Jiaming and Wang, Binghui and Li, Qi and Sun, Kun and Xu, Mingwei and Liu, Zhuotao},
  booktitle={ACM Conference on Computer and Communications Security},
  year={2021}
}

@inproceedings{wang2022bandits,
  title={Bandits for Structure Perturbation-based Black-box Attacks to Graph Neural Networks with Theoretical Guarantees},
  author={Wang, Binghui and Li, Youqi and Zhou, Pan},
  booktitle={IEEE/CVF Conference on Computer Vision and Pattern Recognition},
  year={2022}
}

@inproceedings{wang2023turning,
  title={Turning Strengths into Weaknesses: A Certified Robustness Inspired Attack Framework against Graph Neural Networks},
  author={Wang, Binghui and Pang, Meng and Dong, Yun},
  booktitle={IEEE/CVF Conference on Computer Vision and Pattern Recognition},
  year={2023}
}

@inproceedings{wang2019attacking,
  title={Attacking graph-based classification via manipulating the graph structure},
  author={Wang, Binghui and Gong, Neil Zhenqiang},
  booktitle={ACM Conference on Computer and Communications Security},
  year={2019}
}

@inproceedings{xu2025discrete,
  author       = {Zhe Xu and
                  Ruizhong Qiu and
                  Yuzhong Chen and
                  Huiyuan Chen and
                  Xiran Fan and
                  Menghai Pan and
                  Zhichen Zeng and
                  Mahashweta Das and
                  Hanghang Tong},
  title        = {Discrete-state Continuous-time Diffusion for Graph Generation},
  booktitle    = {Advances in Neural Information Processing Systems},
  year         = {2024},
}

@article{ligraphmaker,
  title={GraphMaker: Can Diffusion Models Generate Large Attributed Graphs?},
  author={Li, Mufei and Kreacic, Eleonora and Potluru, Vamsi K and Li, Pan},
  journal={Transactions on Machine Learning Research},
year={2024}
}

@inproceedings{yang2023directional,
  title={Directional diffusion models for graph representation learning},
  author={Yang, Run and Yang, Yuling and Zhou, Fan and Sun, Qiang},
  booktitle={Advances in Neural Information Processing Systems},
  volume={36},
  pages={32720--32731},
  year={2023}
}

@article{wu2018moleculenet,
  title={{MoleculeNet}: a benchmark for molecular machine learning},
  author={Wu, Zhenqin and Ramsundar, Bharath and Feinberg, Evan N and Gomes, Joseph and Geniesse, Caleb and Pappu, Aneesh S and Leswing, Karl and Pande, Vijay},
  journal={Chemical science},
  volume={9},
  number={2},
  pages={513--530},
  year={2018},
  publisher={Royal Society of Chemistry}
}

@article{polykovskiy2020molecular,
  title={Molecular sets ({MOSES}): a benchmarking platform for molecular generation models},
  author={Polykovskiy, Daniil and Zhebrak, Alexander and Sanchez-Lengeling, Benjamin and others},
  journal={Frontiers in pharmacology},
  volume={11},
  pages={565644},
  year={2020},
  publisher={Frontiers Media SA}
}

@article{brown2019guacamol,
  title={{GuacaMol}: benchmarking models for de novo molecular design},
  author={Brown, Nathan and Fiscato, Marco and Segler, Marwin HS and Vaucher, Alain C},
  journal={Journal of chemical information and modeling},
  volume={59},
  number={3},
  pages={1096--1108},
  year={2019},
  publisher={ACS Publications}
}

@article{wills2020metrics,
  title={Metrics for graph comparison: a practitioner’s guide},
  author={Wills, Peter and Meyer, Fran{\c{c}}ois G},
  journal={Plos one},
  volume={15},
  number={2},
  pages={e0228728},
  year={2020},
  publisher={Public Library of Science San Francisco, CA USA}
}

@inproceedings{simonovsky2018graphvae,
  title={{GraphVAE}: Towards generation of small graphs using variational autoencoders},
  author={Simonovsky, Martin and Komodakis, Nikos},
  booktitle={Artificial Neural Networks and Machine Learning},
  pages={412--422},
  year={2018},
  organization={Springer}
}

@article{ma2018constrained,
  title={Constrained generation of semantically valid graphs via regularizing variational autoencoders},
  author={Ma, Tengfei and Chen, Jie and Xiao, Cao},
  journal={Advances in Neural Information Processing Systems},
  year={2018}
}

@inproceedings{you2018graphrnn,
  title={{GraphRNN}: Generating realistic graphs with deep auto-regressive models},
  author={You, Jiaxuan and Ying, Rex and Ren, Xiang and Hamilton, William and Leskovec, Jure},
  booktitle={International conference on machine learning},
  year={2018}
}

@inproceedings{dai2020scalable,
  title={Scalable deep generative modeling for sparse graphs},
  author={Dai, Hanjun and Nazi, Azade and Li, Yujia and Dai, Bo and Schuurmans, Dale},
  booktitle={International conference on machine learning},
  pages={2302--2312},
  year={2020},
  organization={PMLR}
}

@article{zahirnia2022micro,
  title={Micro and macro level graph modeling for graph variational auto-encoders},
  author={Zahirnia, Kiarash and Schulte, Oliver and Naddaf, Parmis and Li, Ke},
  journal={Advances in Neural Information Processing Systems},
  year={2022}
}

@article{maziarka2020mol,
  title={{Mol-CycleGAN}: a generative model for molecular optimization},
  author={Maziarka, {\L}ukasz and Pocha, Agnieszka and Kaczmarczyk, Jan and Rataj, Krzysztof and Danel, Tomasz and Warcho{\l}, Micha{\l}},
  journal={Journal of Cheminformatics},
  volume={12},
  number={1},
  pages={2},
  year={2020},
  publisher={Springer}
}

@article{madhawa2019graphnvp,
  title={{GraphNVP}: An invertible flow model for generating molecular graphs},
  author={Madhawa, Kaushalya and Ishiguro, Katushiko and Nakago, Kosuke and Abe, Motoki},
  journal={arXiv preprint arXiv:1905.11600},
  year={2019}
}

@inproceedings{zang2020moflow,
  title={{MoFlow}: an invertible flow model for generating molecular graphs},
  author={Zang, Chengxi and Wang, Fei},
  booktitle={Proceedings of the 26th ACM SIGKDD international conference on knowledge discovery \& data mining},
  year={2020}
}

@inproceedings{kuznetsov2021molgrow,
  title={{MolGrow}: A graph normalizing flow for hierarchical molecular generation},
  author={Kuznetsov, Maksim and Polykovskiy, Daniil},
  booktitle={AAAI Conference on Artificial Intelligence},
  year={2021}
}

@inproceedings{luo2021graphdf,
  title={{GraphDF}: A discrete flow model for molecular graph generation},
  author={Luo, Youzhi and Yan, Keqiang and Ji, Shuiwang},
  booktitle={International conference on machine learning},
  pages={7192--7203},
  year={2021},
  organization={PMLR}
}

@article{li2018learning,
  title={Learning deep generative models of graphs},
  author={Li, Yujia and Vinyals, Oriol and Dyer, Chris and Pascanu, Razvan and Battaglia, Peter},
  journal={arXiv preprint arXiv:1803.03324},
  year={2018}
}

@article{liu2018constrained,
  title={Constrained graph variational autoencoders for molecule design},
  author={Liu, Qi and Allamanis, Miltiadis and Brockschmidt, Marc and Gaunt, Alexander},
  journal={Advances in neural information processing systems},
  volume={31},
  year={2018}
}

@inproceedings{kong2021diffwave,
  title={{DiffWave}: A Versatile Diffusion Model for Audio Synthesis},
  author={Kong, Zhifeng and Ping, Wei and Huang, Jiaji and Zhao, Kexin and Catanzaro, Bryan},
  booktitle={International Conference on Learning Representations},
  year={2021},
}

@inproceedings{liu2023audioldm,
  title={{AudioLDM}: Text-to-Audio Generation with Latent Diffusion Models},
  author={Liu, Haohe and Chen, Zehua and Yuan, Yi and Mei, Xinhao and Liu, Xubo and Mandic, Danilo and Wang, Wenwu and Plumbley, Mark D},
  booktitle={International Conference on Machine Learning},
  pages={21450--21474},
  year={2023},
  organization={PMLR}
}

@article{dhariwal2021diffusion,
  title={Diffusion models beat {GANs} on image synthesis},
  author={Dhariwal, Prafulla and Nichol, Alexander},
  journal={Advances in neural information processing systems},
  
  volume={34},
  pages={8780--8794},
  year={2021}
}

@article{ho2022video,
  title={Video diffusion models},
  author={Ho, Jonathan and Salimans, Tim and Gritsenko, Alexey and Chan, William and Norouzi, Mohammad and Fleet, David J},
  journal={Advances in Neural Information Processing Systems},
  volume={35},
  pages={8633--8646},
  year={2022}
}

@inproceedings{sohl2015deep,
  title={Deep unsupervised learning using nonequilibrium thermodynamics},
  author={Sohl-Dickstein, Jascha and Weiss, Eric and Maheswaranathan, Niru and Ganguli, Surya},
  booktitle={International conference on machine learning},
  year={2015}
}

@inproceedings{niu2020permutation,
  title={Permutation invariant graph generation via score-based generative modeling},
  author={Niu, Chenhao and Song, Yang and Song, Jiaming and Zhao, Shengjia and Grover, Aditya and Ermon, Stefano},
  booktitle={International Conference on Artificial Intelligence and Statistics},
  year={2020}
}

@inproceedings{jo2022score,
  title={Score-based generative modeling of graphs via the system of stochastic differential equations},
  author={Jo, Jaehyeong and Lee, Seul and Hwang, Sung Ju},
  booktitle={International conference on machine learning},
  year={2022}
}

@inproceedings{xu2022geodiff,
  title={{GeoDiff}: A Geometric Diffusion Model for Molecular Conformation Generation},
  author={Xu, Minkai and Yu, Lantao and Song, Yang and Shi, Chence and Ermon, Stefano and Tang, Jian},
  booktitle={International Conference on Learning Representations},
 year={2022},
}

@inproceedings{madry2018towards,
  author       = {Aleksander Madry and
                  Aleksandar Makelov and
                  Ludwig Schmidt and
                  Dimitris Tsipras and
                  Adrian Vladu},
  title        = {Towards Deep Learning Models Resistant to Adversarial Attacks},
  booktitle    = {International Conference on Learning Representations},
  year         = {2018},
}

@article{austin2021structured,
  title={Structured denoising diffusion models in discrete state-spaces},
  author={Austin, Jacob and Johnson, Daniel D and Ho, Jonathan and Tarlow, Daniel and Van Den Berg, Rianne},
  journal={Advances in Neural Information Processing Systems},
  volume={34},
  pages={17981--17993},
  year={2021}
}

@inproceedings{jin2018junction,
  title={Junction tree variational autoencoder for molecular graph generation},
  author={Jin, Wengong and Barzilay, Regina and Jaakkola, Tommi},
  booktitle={International conference on machine learning},
  pages={2323--2332},
  year={2018},
  organization={PMLR}
}

@inproceedings{shi2020graphaf,
  title={{GraphAF}: a Flow-based Autoregressive Model for Molecular Graph Generation},
  author={Shi, Chence and Xu, Minkai and Zhu, Zhaocheng and Zhang, Weinan and Zhang, Ming and Tang, Jian},
  booktitle={International Conference on Learning Representations},
year={2020},
}

@inproceedings{jin2020hierarchical,
  title={Hierarchical generation of molecular graphs using structural motifs},
  author={Jin, Wengong and Barzilay, Regina and Jaakkola, Tommi},
  booktitle={International conference on machine learning},
  pages={4839--4848},
  year={2020},
  organization={PMLR}
}

@inproceedings{vignac2023digress,
  title={{DiGress}: Discrete Denoising diffusion for graph generation},
  author={Vignac, Clement and Krawczuk, Igor and Siraudin, Antoine and Wang, Bohan and Cevher, Volkan and Frossard, Pascal},
  booktitle={The Eleventh International Conference on Learning Representations},
year={2023}
}

@inproceedings{kong2023autoregressive,
  title={Autoregressive diffusion model for graph generation},
  author={Kong, Lingkai and Cui, Jiaming and Sun, Haotian and Zhuang, Yuchen and Prakash, B Aditya and Zhang, Chao},
  booktitle={International conference on machine learning},
  year={2023}
}

@inproceedings{chen2023efficient,
  title={Efficient and degree-guided graph generation via discrete diffusion modeling},
  author={Chen, Xiaohui and He, Jiaxing and Han, Xu and Liu, Li-Ping},
  booktitle={International Conference on Machine Learning},
  year={2023}
}

@inproceedings{gruver2024protein,
  author       = {Nate Gruver and
                  Samuel Stanton and
                  Nathan C. Frey and
                  Tim G. J. Rudner and
                  Isidro H{\"{o}}tzel and
                  Julien Lafrance{-}Vanasse and
                  Arvind Rajpal and
                  Kyunghyun Cho and
                  Andrew Gordon Wilson},
  title        = {Protein Design with Guided Discrete Diffusion},
  booktitle    = {Advances in Neural Information Processing Systems},
  year         = {2023},
}

@inproceedings{yi2024graph,
  author       = {Kai Yi and
                  Bingxin Zhou and
                  Yiqing Shen and
                  Pietro Li{\'{o}} and
                  Yuguang Wang},
  title        = {Graph Denoising Diffusion for Inverse Protein Folding},
  booktitle    = {Advances in Neural Information Processing Systems},
  year         = {2023},
}

@inproceedings{chou_how_2023,
	address = {Vancouver, BC, Canada},
	title = {How to {Backdoor} {Diffusion} {Models}?},
	booktitle = {{IEEE}/{CVF} {Conference} on {Computer} {Vision} and {Pattern} {Recognition}},
	author = {Chou, Sheng-Yen and Chen, Pin-Yu and Ho, Tsung-Yi},
	month = jun,
	year = {2023}
}

@inproceedings{liu_generative_2023,
	address = {Macau, SAR China},
	title = {Generative {Diffusion} {Models} on {Graphs}: {Methods} and {Applications}},
	language = {en},
	booktitle = {International Joint Conferences on Artificial Intelligence},
	author = {Liu, Chengyi and Fan, Wenqi and Liu, Yunqing and Li, Jiatong and Li, Hang and Liu, Hui and Tang, Jiliang and Li, Qing},
	month = aug,
	year = {2023},
	pages = {6702--6711},
}

@inproceedings{kohler2020equivariant,
  title={Equivariant flows: exact likelihood generative learning for symmetric densities},
  author={K{\"o}hler, Jonas and Klein, Leon and No{\'e}, Frank},
  booktitle={International conference on machine learning},
  pages={5361--5370},
  year={2020},
  organization={PMLR}
}

@INPROCEEDINGS{chen_trojdiff:_2023,
  author={Chen, Weixin and Song, Dawn and Li, Bo},
  booktitle={IEEE/CVF Conference on Computer Vision and Pattern Recognition}, 
  title={TrojDiff: Trojan Attacks on Diffusion Models with Diverse Targets}, 
  year={2023},
  }

@inproceedings{zhang_backdoor_2021,
	address = {Virtual Event Spain},
	title = {Backdoor {Attacks} to {Graph} {Neural} {Networks}},
	booktitle = {{ACM} {Symposium} on {Access} {Control} {Models} and {Technologies}},
	author = {Zhang, Zaixi and Jia, Jinyuan and Wang, Binghui and Gong, Neil Zhenqiang},
	year = {2021}
}

@inproceedings{xi2021graph,
  title={Graph backdoor},
  author={Xi, Zhaohan and Pang, Ren and Ji, Shouling and Wang, Ting},
  booktitle={USENIX Security},
  year={2021}
}

@inproceedings{yang2024distributed,
  title={Distributed backdoor attacks on federated graph learning and certified defenses},
  author={Yang, Yuxin and Li, Qiang and Jia, Jinyuan and Hong, Yuan and Wang, Binghui},
  booktitle={ACM Conference on Computer and Communications Security},
  year={2024}
}

@inproceedings{ho_denoising_2020,
  author       = {Jonathan Ho and
                  Ajay Jain and
                  Pieter Abbeel},
  title        = {Denoising Diffusion Probabilistic Models},
  booktitle    = {Advances in Neural Information Processing Systems},
  year         = {2020},
}
\bibliographystyle{tmlr}

\clearpage
\appendix

\section{Proofs}
\label{app:proofs}

The below proofs \ref{app:backdoor_graphpost}-\ref{app:limitbackdoor} derive the three properties ({\bf P1}-{\bf P3}) required in Section \ref{sec:background} for our setting.
\begin{align*}
& \textrm{\bf P1: forward distribution } q(G_B^t | G_B) \\
& \textrm{\bf P2: limit distribution } \lim_{t \rightarrow \infty} q(G_B^t) \\
& \textrm{\bf P3: reverse denoising distribution } 
q(G_B^{t-1}|G_B^t, G_B)
\end{align*}

\subsection{\texorpdfstring{Deriving $q(G_B^t \mid G_B)$}
{Deriving q(GB(t) | GB)}}
\label{app:backdoor_graphpost}

We derive  $q(\mE_B^{t}|\mE_B)$ for simplicity as it is identical to derive  $q(\mX_B^{t}|\mX_B)$.  
Recall 
{
\begin{align*}
& {\mE}_B^t | {\mE}_B \sim \boldsymbol{\mE^t} \odot (1-\mM_E)  + \mE_s \odot \mM_E. \\
& \mE_B^t|\mE_B^{t-1} \sim \mE^{t-1}\mQ_{E_B}^{t} \odot(1-\mM_E) +\mE_s \odot \mM_E
\end{align*}
}%

Due to the properties of Markov chain and  $q(\mE_B^t|\mE_B^{t-1})$, following existing discrete diffusion models \citep{austin2021structured}, one can marginalize out the intermediate steps and derive below:
{
\begin{align*}
\label{eqn:forwardlimi_backdoor_appendix}
& q(\mE_B^{t}|\mE_B) = \mE\bar{Q}_{E_B}^{t} \odot(1-\mM_E) +\mE_s \odot \mM_E
\end{align*}
}%

\subsection{\texorpdfstring
{Deriving $q(G_B^{t-1} \mid G_B^t, G_B)$}
{Deriving q(GB(t-1) | GB(t), GB)}}
\label{app:backdoor_post2}

We derive $ q(\mE_B^{t-1}|\mE_B^t, \mE_B)$ for simplicity as it is identical to derive $ q(\mX_B^{t-1}|\mX_B^t, \mX_B)$.  
{
\begin{align*}
& q(\mE_B^{t-1}|\mE_B^t, \mE_B) \\
= ~&  q(\mE_B^t | \mE_B^{t-1}, \mE_B) ~ q(\mE_B^{t-1}|\mE_B)\\
= ~&   q(\mE_B^t | \mE_B^{t-1}) ~ q(\mE_B^{t-1}|\mE_B)
\propto ~ 
q(\mE_B^{t-1}|\mE_B^{t})  ~ q(\mE_B^{t-1}|\mE_B) \\
= ~& \Big(\mE^t(Q_{E_B}^{t})'\odot(1-\mM_E) +\mE_s \odot \mM_E \Big) 
\odot \Big(\mE \bar{Q}_{E_B}^{t-1}\odot(1-\mM_E) +\mE_s \odot \mM_E \Big) \\
= ~& \mE^t(Q_{E_B}^{t})' \odot \mE \bar{Q}_{E_B}^{t-1} \odot(1-\mM_E) +\mE_s \odot \mM_E,
\end{align*}
}% 

where the first and third equations use the Bayesian rule,  the second equation uses the Markov property, the fourth equation  uses the definition of $\bar{Q}_{E_B}$ in the opposite direction,  
and the last equation we use that $(1-\mM_E) \odot  \mM_E = 0$, $(1 - \mM_E) \odot (1 - \mM_E) =  (1 - \mM_E)$, and $\mM_E \odot \mM_E =  \mM_E $.

\subsection{Deriving  Equation \ref{eqn:backdoor_limitdist}}
\label{app:limitbackdoor}

Recall $\mQ^t_{X_B} = \alpha^t \mI + (1-\alpha^t)~\bm 1_a \bm \vm_{X_B}'$ and $ 
\mQ^t_{E_B} = \alpha^t \mI + (1-\alpha^t) ~\bm 1_d \bm \vm_{E_B}'$.
Then we 
show the limit probability of jumping from any state to a state $j$ is proportional to the marginal probability of category $j$. Formally, 
$$ \lim_{T \to \infty} (\bar \mQ^T_{X_B}, \bar \mQ^T_{E_B}) {\bf e}_i = (\vm_{X_B}, \vm_{E_B}), \quad \forall i.$$

We ignore subscripts $a,d$, $X_B$, and $E_B$ for description simplicity. 
First, we show the square of the row-column product $(\bm 1 \bm \vm')^2 = \bm 1 \bm \vm' \bm 1 \bm \vm' = \bm 1 \bm \vm'$, where the column-row product $\vm' \bm 1 = 1$, as $\vm$ is a probability vector.  

Next, we prove via induction that: $\bar \mQ^t = \bar \alpha^t \mI + \bar \beta^t \bm 1 \bm \vm'$ for $\bar\alpha^t = \prod_{\tau=1}^t \alpha^\tau$ and $\bar \beta^t = 1 - \bar\alpha^t$. 

\noindent {\bf Step I: Base case.} When $t=1$,
$\bar \mQ^1 = \mQ^1 = \alpha^1 \mI + \beta^1 \bm 1 \bm \vm' = \bar \alpha^1 \mI + \bar \beta^1 \bm 1 \bm \vm'$, satisfying the base case. 

\noindent {\bf Step II: Inductive Hypothesis.} 
Assume $t=k$,  $\bar \mQ^k = \bar \alpha^k \mI + \bar \beta^k \bm 1 \bm \vm'$ for $\bar\alpha^k = \prod_{\tau=1}^k \alpha^\tau$ and $\bar \beta^k = 1 - \bar\alpha^k$.

\noindent {\bf Step III: Inductive Step.} We prove that $\bar \mQ^{k+1} = \bar \alpha^{k+1} \mI + \bar \beta^{k+1} \bm 1 \bm \vm'$ for $\bar\alpha^{k+1} = \prod_{\tau=1}^{k+1} \alpha^\tau$ and $\bar \beta^{k+1} = 1 - \bar\alpha^{k+1}$.
The detail is shown below: 
\begin{align*}
& \bar \mQ^{k+1} = \bar \mQ^{k} \mQ^{k+1} \\
& = (\bar \alpha^k \mI + \bar \beta^k \bm 1 \bm \vm') ~ (\alpha^{k+1} \mI + \beta^{k+1} ~\bm 1 \bm \vm') \\
& = \bar \alpha^k \alpha^{k+1} \mI + (\bar \alpha^k  \beta^{k+1} + \bar \beta^k  \alpha^{k+1}) \bm 1 \bm \vm' + \bar \beta^k  \beta^{k+1} \bm 1 \bm \vm' \bm 1 \bm \vm' \\
& = \bar \alpha^{k+1} \mI + \big(\bar \alpha^k (1-\alpha^{k+1}) + (1-\bar \alpha^k) \alpha^{k+1} 
+ (1-\bar \alpha^k) (1-\alpha^{k+1}) \big)  \bm 1 \bm \vm' \\
& =  \bar \alpha^{k+1} \mI  + (1-\bar \alpha^{k+1}) \bm 1 \bm \vm' 
\end{align*}

As $T \rightarrow \infty$, $\bar \alpha^{T} \rightarrow 0$. Hence $\lim_{T \to \infty} \bar \mQ^T = 1 \bm \vm'$, where all rows are $\bm \vm'$.  
Thus, for any base vector ${\bf e}_i$, $ \lim_{T \to \infty} \bar \mQ^T {\bf e}_i = \vm$.

\section{Permutation Equivariance and Exchangeability}
\label{app_proof}

 DiGress establishes permutation 
equivariance
 and distribution exchangeability. Since backdoored DiGress adopts the same network architecture and loss type, it also preserves these properties. We include the proofs here for completeness.

\subsection{Proof of Theorem \ref{thm:perinvariant}}
\label{app:thm:perinvariant}

\begin{theorem}[Backdoored DiGress is Permutation Equivariant]
\label{thm:perinvariant} 
Let $G^t = (\mX^t, \mE^t) $ be an intermediate noised (clean or backdoored) graph, and $\pi(G^t) = (\pi(\mX^t), \pi(\mE^t))$ be its permutation. 
Backdoored DiGress is permutation 
equivariant,
i.e., $p_{\theta_B}(\pi(G^t)) = \pi (p_{\theta_B}(G^t))$.
\end{theorem}

We need to prove that: i) the neural network building blocks are permutation 
equivariant;
and ii) the objective function (i.e., the training loss) is also permutation invariant. 

\noindent {\bf Proving i):} DiGress uses three types of blocks: 

- 1) spectral and structural features (e.g., eigenvalues of the graph Laplacian and cycles in the graph to improve the network expressivity); 

- 2) graph transformer layers (consisting of graph self-attention and fully connected multi-layer perceptron);  

- 3) layer-normalization. 

 DiGress proves that these blocks are permutation 
equivariant.
Backdoored DiGress uses the same network architecture as DiGress and hence is also permutation 
equivariant.

\noindent {\bf Proving ii):} Backdoored DiGress optimizes the  cross-entropy loss   on clean graphs $\{G=(\mX, \mE)\}$ and backdoored graphs $\{G^B = (\mX_B, \mE_B)\}$ to learn the model $\theta_B$: 
{

\small
\begin{align*}
& \min_{\theta_B} 
\mathcal{L}(\{G\}, \{G_B\}; \theta_B) 
=  
\sum_{\{G=(\mX, \mE)\}} \big( l_{CE}(\mX, \hat \vp^X) + l_{CE}(\mE, \hat \vp^E) \big) 
+ \sum_{\{G^B = (\mX_B, \mE_B)\}} \big( l_{CE}(\mX_B, \hat \vp^{X_B}) +  l_{CE}(\mE_B, \hat \vp^{E_B}) \big) 
\end{align*}
}

For a clean graph $G$ or a backdoored graph $G_B$, its associated cross-entropy loss can be decomposed to be the sum of the loss of individual nodes and edges. For instance,   $l_{CE}(\mX, \hat \vp^X) = \sum_{1 \leq i \leq n} l_{CE}(x_i, \hat p^X_i)$,  $l_{CE}(\mE_B, \hat \vp^{E_B}) = \sum_{1 \leq i, j \leq n} l_{CE}(e_{B,ij}, \hat p^E_{B,ij})$. 

Hence, the total loss on the clean and backdoored graphs does not change with any node permutation $\pi$. That is, 
$$\mathcal{L}(\{\pi(G)\}, \{\pi(G_B)\}; \theta_B) = 
\mathcal{L}(\{G\}, \{G_B\}; \theta_B).$$

\subsection{Proof of Theorem \ref{thm:exchangeability}}
\label{app:thm:exchangeability}

\begin{theorem}[Backdoored DiGress Produces Exchangeable Distributions]
\label{thm:exchangeability}
Backdoored DiGress generates graphs with node features $\mX$ and edges $\mE$ that satisfy $P(\mX, \mE) = P(\pi(\mX), \pi(\mE))$ for any permutation $\pi$.
\end{theorem}

The proof builds on the result in \citet{xu2022geodiff}: 
\begin{proposition}[Adapted from \citet{xu2022geodiff}]
\label{prop:exchangeability}
Let $\mathcal{C}$ be a particle. If,

i) a distribution \( p(\mathcal{C}^T) \) is invariant under the transformation $T_g$ of a group element \( g\), i.e.,  \( p(\mathcal{C}^T) =  p(T_g(\mathcal{C}^T)) \);

ii) the Markov transitions \( p(\mathcal{C}^{t-1} \mid \mathcal{C}^t) \) are equivariant, i.e., \( p (\mathcal{C}^{t-1} \mid \mathcal{C}^t) =  p(T_g(\mathcal{C}^{t-1}) \mid T_g(\mathcal{C}^t)) \),

then the density \( p_\theta(\mathcal{C}^0) \) is also invariant under the transformation \( T_g \), i.e., \( p_\theta (\mathcal{C}^0) =  p_\theta(T_g(\mathcal{C}^0))\). 
\end{proposition}

We apply Proposition \ref{prop:exchangeability} to our setting:

First, the clean or backdoored limit distribution $p(G^T)$ or $p(G_B^T)$ is the product of independent  and  identical distribution on each node and edge. It is thus permutation invariant and satisfies condition i).

    Second, the denoising network $p_{\theta_B}$ in backdoored DiGress is permutation equivariant (Theorem~\ref{thm:perinvariant}). Moreover, the network prediction $\hat p_{\theta_B}  (G)\to p_{\theta_B} (G^{t-1} | G^t) = \sum_G q(G^{t-1}, G | G^t) \hat p_{\theta_B}(G)$ defining the transition probabilities is equivariant to joint permutations of $\hat p_{\theta_B}(G)$ and $G^t$, and so to the  joint permutations of $\hat p_{\theta_B}(G_B)$ and $G_B^t$. Thus,  condition ii) is also satisfied. 

Together, backdoored DiGress generates the graph  with node features $\mX$ and edges $\mE$ that satisfy $P(\mX, \mE) = P(\pi(\mX), \pi(\mE))$ for any permutation $\pi$, meaning the generated graphs are exchangeable.

\begin{algorithm}[!t]
\caption{Backdoored  DiGress Training}
\textbf{Input:} Training graphs $\mathcal{G}_{tr}$, poison rate $p\%$, subgraph trigger $G_s=(\mX_s, \mE_s)$, model parameter $\theta_B$, and transition matrices $\{{Q}^t_{X},{Q}^t_{E},{Q}^t_{X_B},{Q}^t_{E_B}\}$.

\textbf{Preprocess:} Sample $p\%$ of $\mathcal{G}_{tr}$ and inject  $G_s$ to obtain the backdoored graphs $\mathcal{G}_B$; clean graphs $\mathcal{G}_C = \mathcal{G}_{tr} \setminus \mathcal{G}_B$  

\begingroup
\let\AND\relax
\begin{algorithmic}[1]
\STATE Sample $G = (\mX, \mE) \sim \{\mathcal{G}_B, \mathcal{G}_C\}$
\STATE Sample $t \sim Uniform(1, \cdots, T)$
\IF{$G \in \mathcal{G}_B$}
\STATE 
// {Sample a noisy backdoored graph}

Sample $G^t \sim (\mX \bar{Q}_{X_B}^{t}\odot(1-\mM_X) + \mX_s \odot \mM_X) \times  (\mE \bar{Q}_{E_B}^{t}\odot(1-\mM_E) +\mE_s \odot \mM_E)  $  
\ELSE
\STATE 
Sample $G^t \sim \mX \bar \mQ_X^t \times \mE \bar \mQ_E^t  $ // {Sample a noisy clean graph}  
\ENDIF

\STATE $\hat{\bf p}^X, \hat{\bf p}^E \gets p_{\theta_B}(G^t)$ // {Forward pass}
\STATE 
// {Minimize the cross-entropy loss}

$\textrm{optimizer.step}(l_{CE}(\hat{\bf p}^X, \mX) + l_{CE}(\hat{\bf p}^E, \mE))$ 
\end{algorithmic}
\endgroup
\label{alg:backdoortraining}
\end{algorithm}

\begin{algorithm}[!t]
\caption{Sampling from Backdoored DiGress}
\textbf{Input:} 
Trained model $p_{\theta_B}$, timestep $T$, marginal distributions $\{ \vm_{X}^n, \vm_{E}^n$, $\vm_{X_B}^n, \vm_{E_B}^n\} $ for all graph sizes $n$. 

\begingroup
\let\AND\relax
\begin{algorithmic}[1]
\STATE Sample a graph size $n$ from training data distribution
\IF{Generating a clean sample}
\STATE Sample $G^T \sim q_X(\vm_{X}^n) \times q_E(\vm_{E}^n)$ 
\ELSE
\STATE Sample $G^T \sim q_X(\vm_{X_B}^n) \times q_E(\vm_{E_B}^n)$ 
\ENDIF
\FOR{$t = T$ to $1$}
\STATE 
{Forward pass:} $\hat{\bf p}^X, \hat{\bf p}^E \gets p_{\theta_B}(G^t)$ 
\STATE 
Compute node posterior: $p_{\theta_B}(x^{t-1}_i|G^t) \gets 
\sum_x q(x_i^{t-1}\mid x_i^t, x_i=x)\hat p_i^X(x)$,
$i \in 1, \dots, n$ 

\STATE 
Compute edge posterior: 
$p_{\theta_B}(e^{t-1}_{ij}|G^t) \gets
\sum_e q(e_{ij}^{t-1}\mid e_{ij}^t, e_{ij}=e)\hat p_{ij}^E(e)$,
$i, j \in 1, \dots, n$
\STATE Generate graph from the categorical distribution: $G^{t-1} \sim \prod_i p_{\theta_B}(x^{t-1}_i|G^t) \prod_{i,j} p_{\theta_B}(e^{t-1}_{ij}|G^t)$ 
\ENDFOR
\STATE \textbf{return} $G^0$
\end{algorithmic}
\endgroup
\label{alg:backdoorsamping}
\end{algorithm}

\section{Experiments}
\subsection{Dataset description}
\label{app:datasets}

\noindent {\bf QM9}: It is a molecule dataset with 4 distinct elements and 5 bond types. The maximum number of heavy atoms in a graph is 9.

\noindent {\bf Molecular Sets (MOSES):} 
It is specially designed to evaluate generative models for molecular graph generation. 
MOSES consists of molecular structures represented in the SMILES format.
The dataset contains 1.9M+ unique molecules derived from the ZINC Clean Leads dataset, ensuring the molecules are drug-like and chemically realistic.

\noindent {\bf GuacaMol:} It is a benchmark suite specifically designed for evaluating generative models in molecular discovery. GuacaMol includes a collection of molecules from the ChEMBL database, a large database of bioactive molecules with drug-like properties.
 The dataset contains 1.3 million drug-like molecules in the SMILES format. 

\noindent {\bf Training and testing:} 
On QM9, we use 100k molecules for training, and 13k for evaluating the attack effectiveness and utility. On MOSES, we use 1.58M graphs for training and 176k molecules for testing. On GuacaMol, 200k molecules are used for training and 40k molecules for testing.

\subsection{Network architecture}

    We use the original DiGress network architecture, which consists of 9 graph transformer layers for QM9, and 12 graph transformer layers for GuacaMol and MOSES.

\subsection{More attack results without defense}
\label{app:results}

\subsubsection{Visualizing generated graphs}
Figures \ref{fig:dynamics_clean} and \ref{fig:dynamics} respectively illustrate example generated clean graphs and backdoored graphs on the three molecular datasets---clean graphs are valid, while backdoored ones are not.

\begin{figure*}[!t]
    \centering
    \begin{minipage}{0.95\linewidth}
        \centering
        \includegraphics[width=0.7\linewidth]{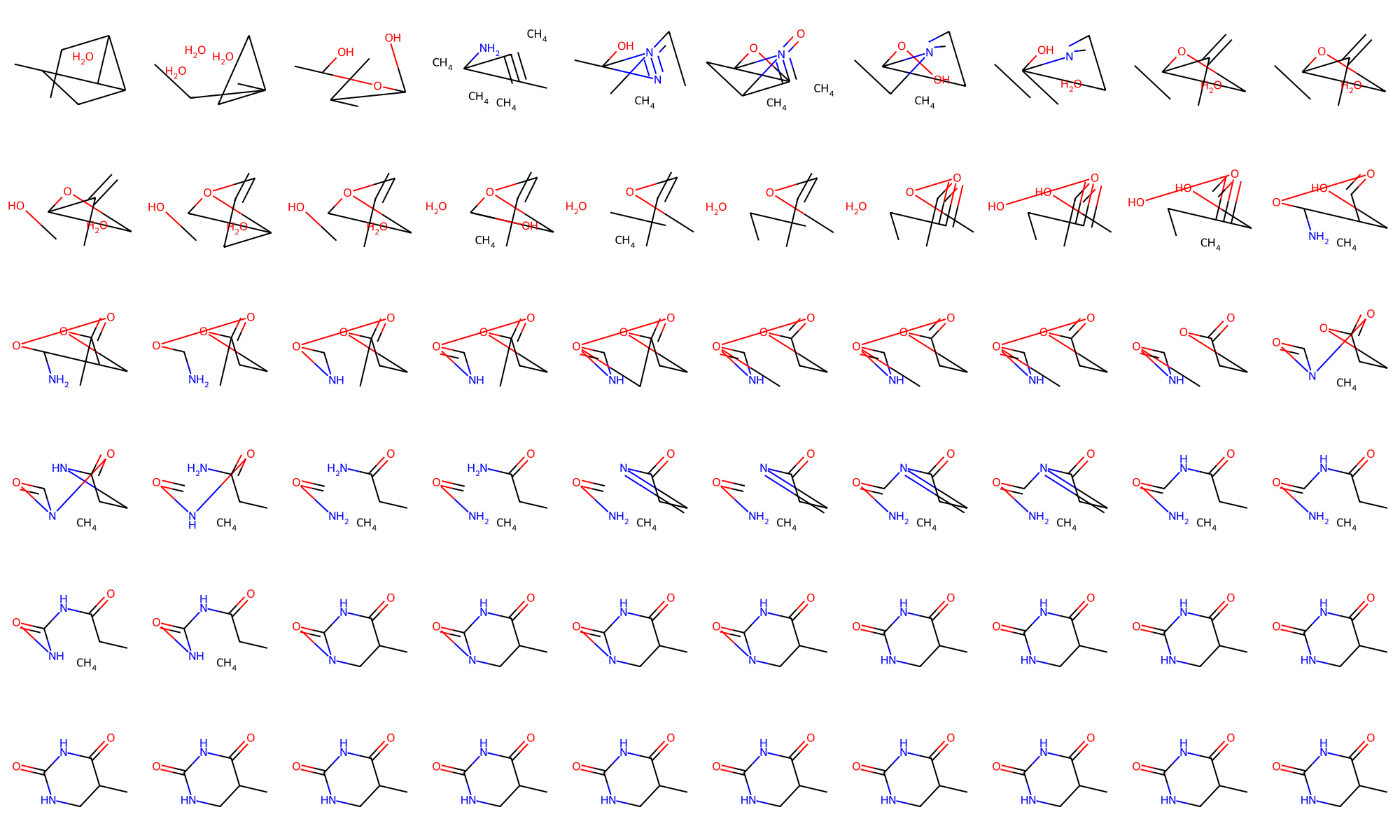}
        \subcaption{QM9-clean}
    \end{minipage}%
    
    \begin{minipage}{0.95\linewidth}
        \centering
        \includegraphics[width=0.7\linewidth]{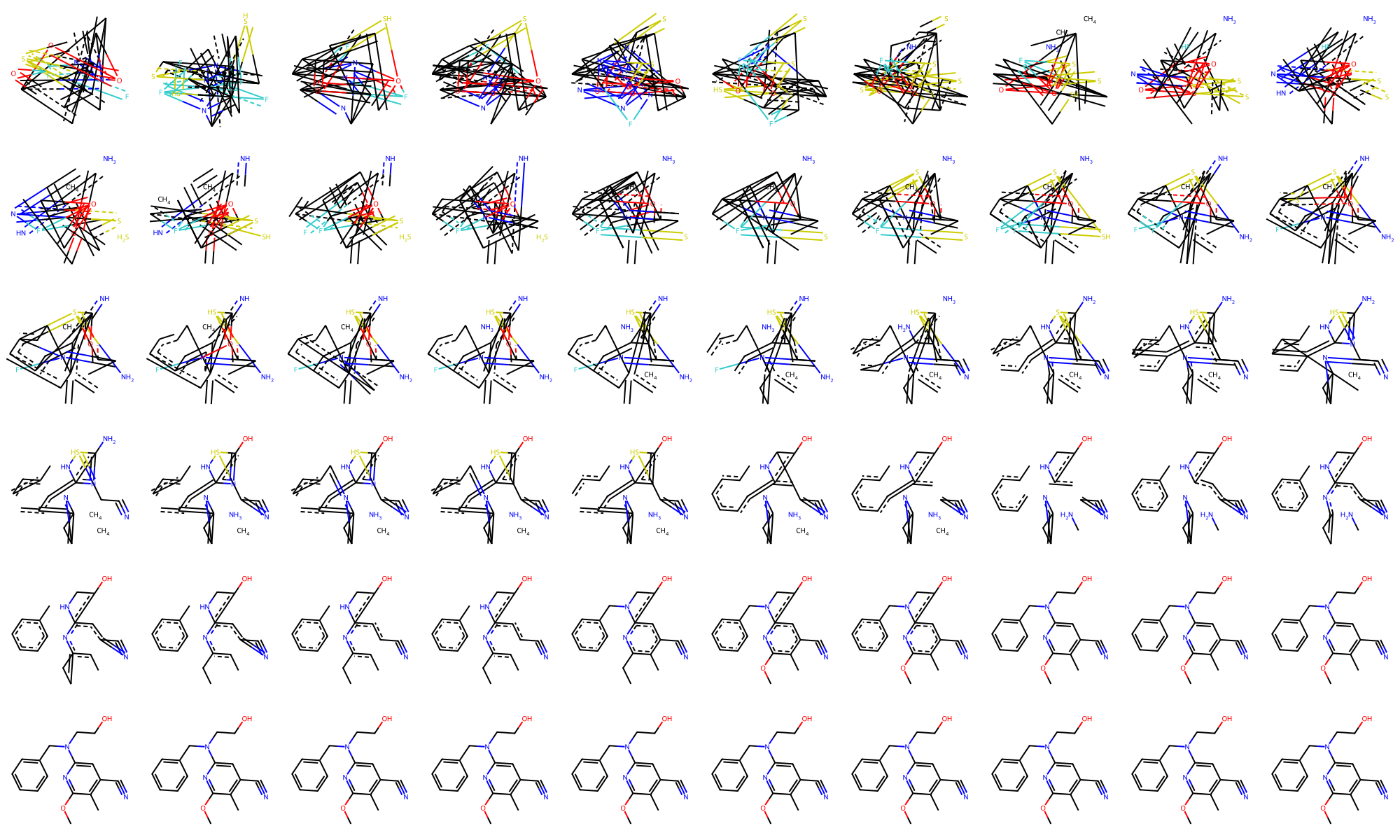}
        \subcaption{MOSES-clean}
    \end{minipage}%
    
    \begin{minipage}{0.95\linewidth}
        \centering
        \includegraphics[width=0.7\linewidth]{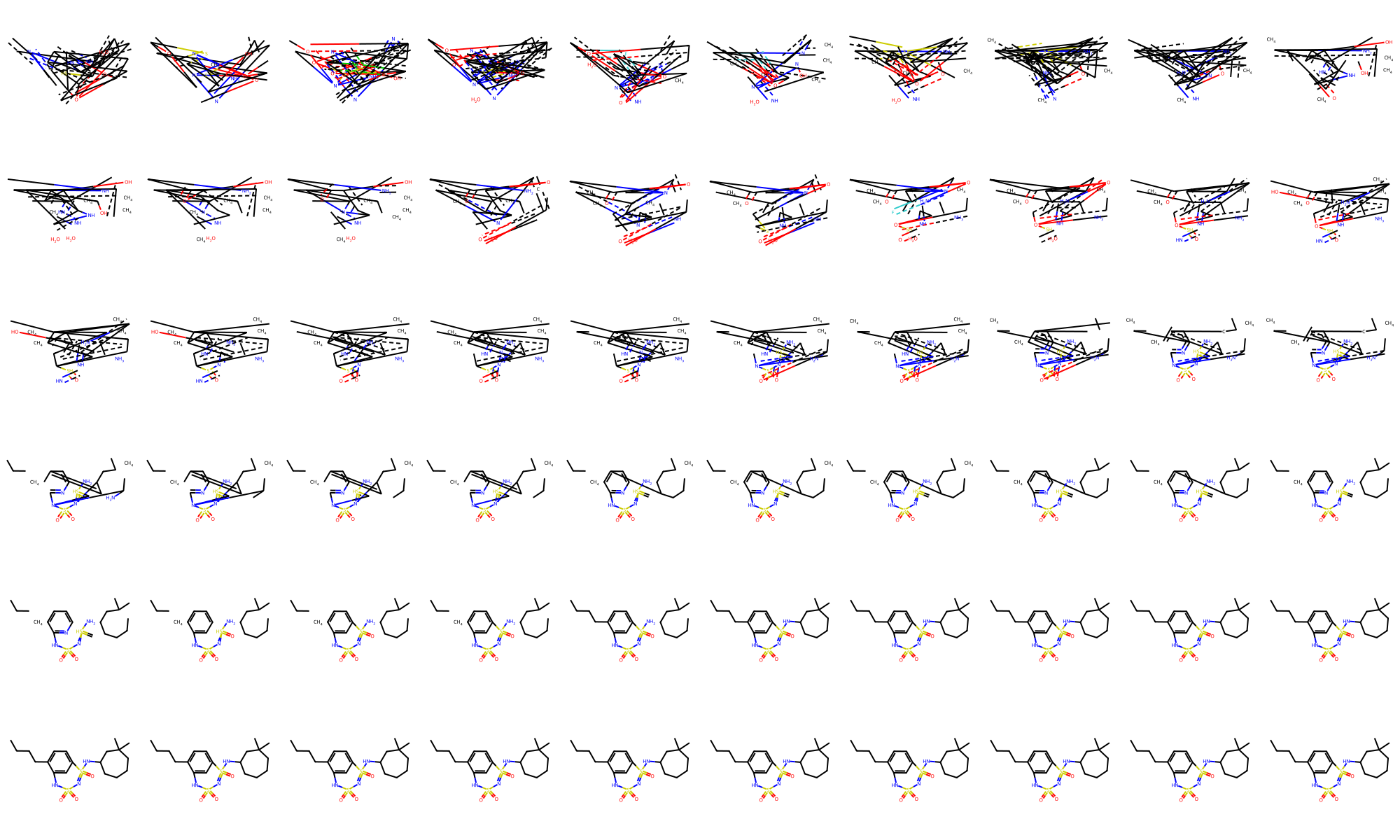}
        \subcaption{GuacaMol-clean}
    \end{minipage}
    
    \caption{Examples of generated clean graphs.}
    \label{fig:dynamics_clean}
\end{figure*}

\begin{figure*}[!t]
    \centering

    \begin{minipage}{0.95\linewidth}
        \centering
        \includegraphics[width=0.7\linewidth]{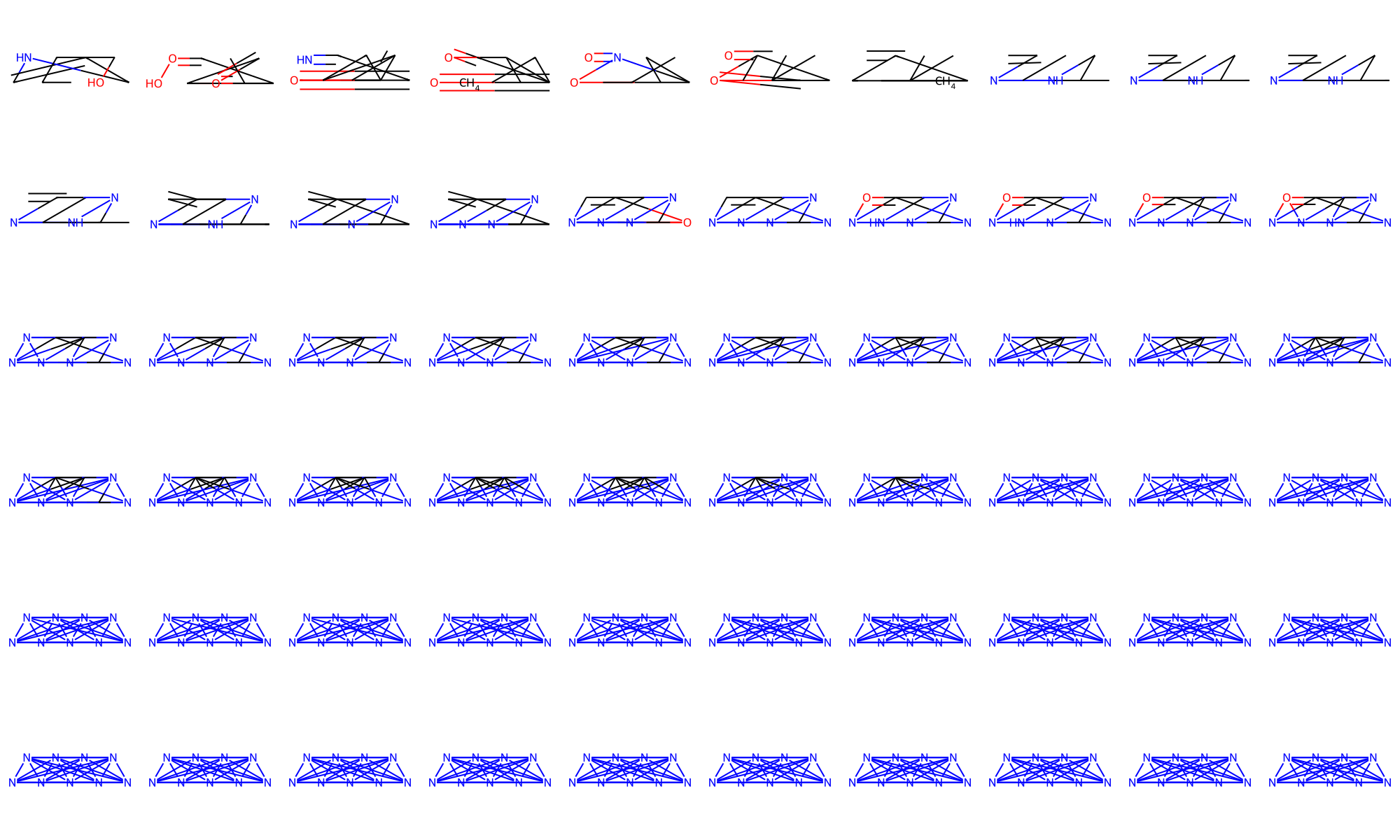}
        \subcaption{QM9-backdoored}
    \end{minipage}%
    
    \begin{minipage}{0.95\linewidth}
        \centering
        \includegraphics[width=0.7\linewidth]{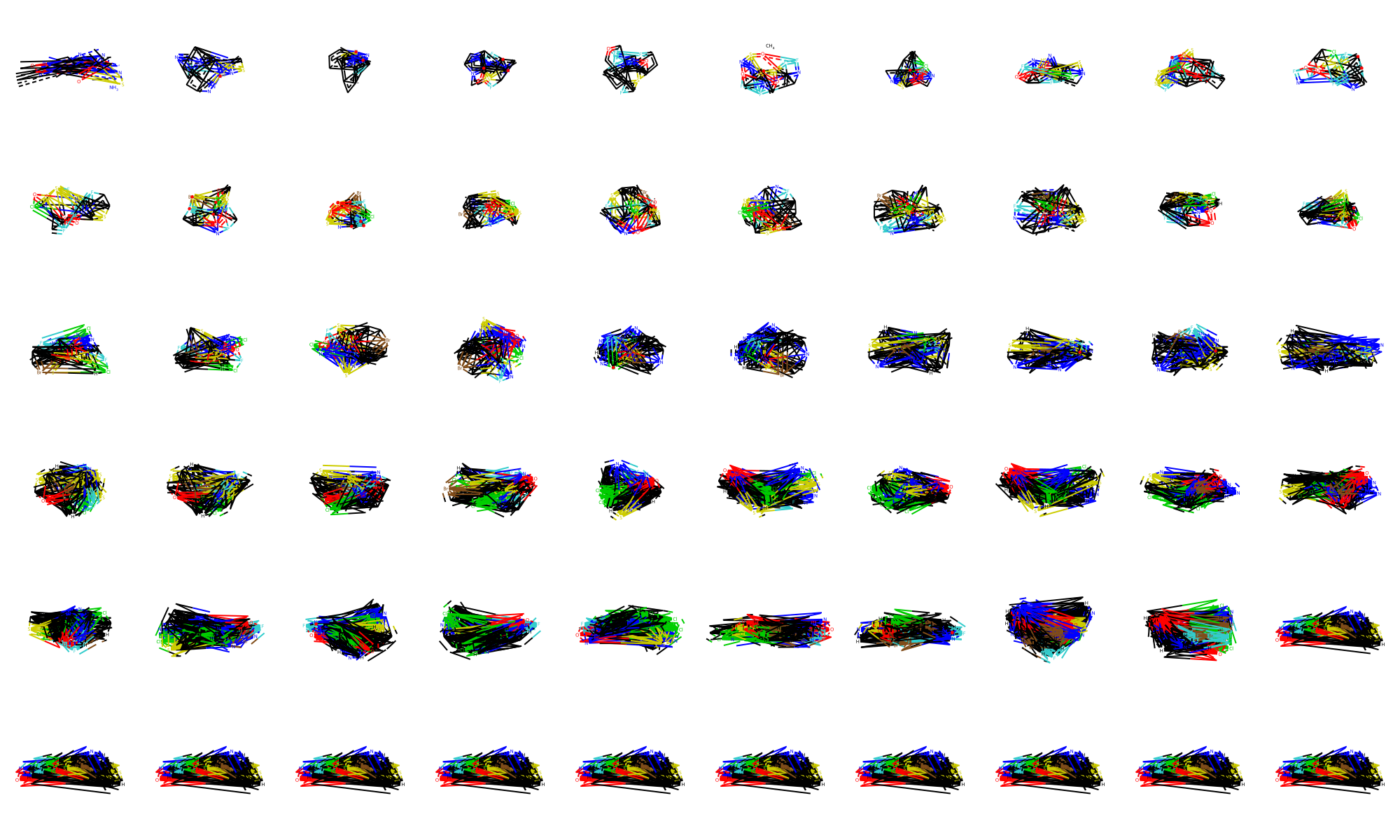}
        \subcaption{MOSES-backdoored}
    \end{minipage}%
    
    \begin{minipage}{0.95\linewidth}
        \centering
        \includegraphics[width=0.7\linewidth]{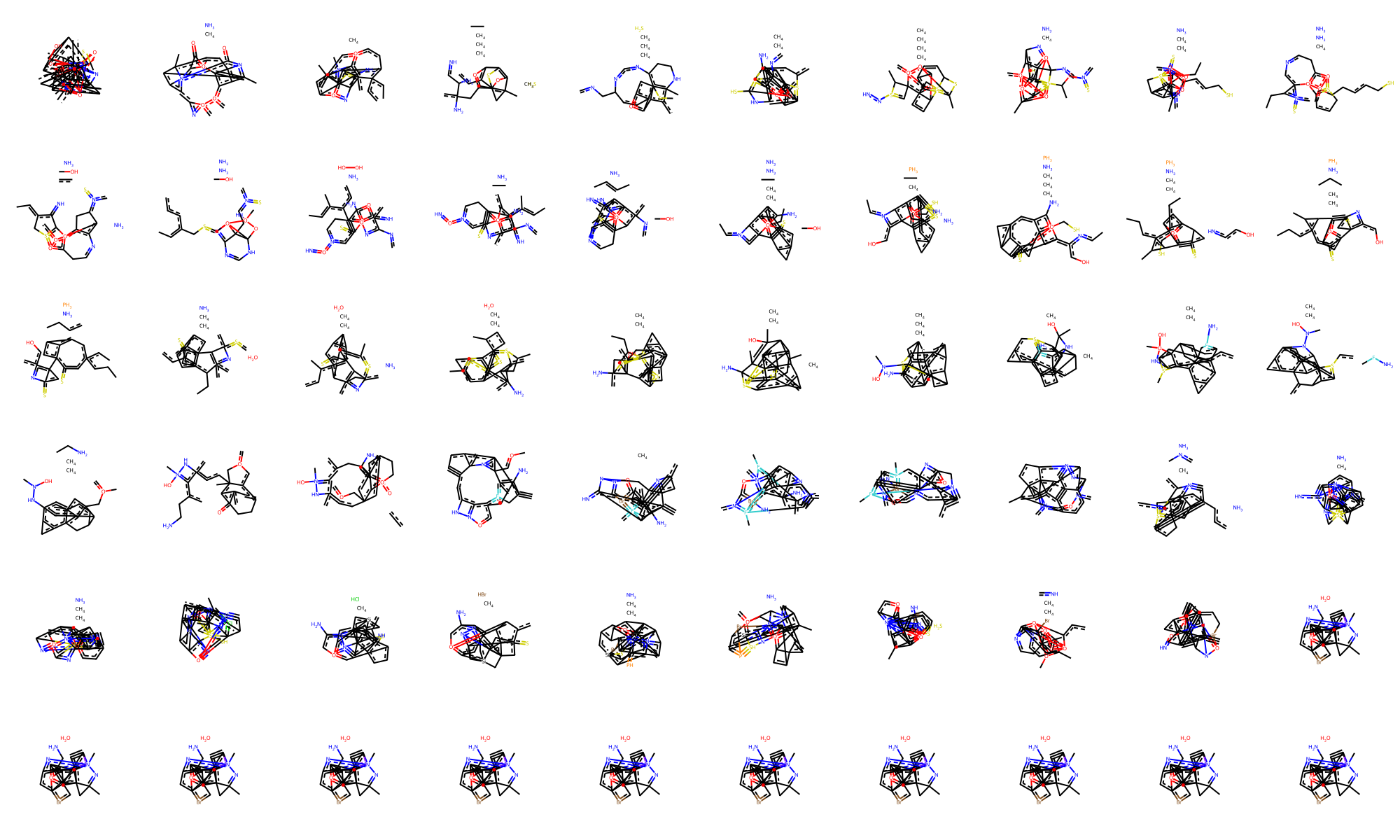}
        \subcaption{GuacaMol-backdoored}
    \end{minipage}
    \caption{Examples of generated backdoored graphs.}
    \label{fig:dynamics}
\end{figure*}

\subsubsection{Attack results on a large graph dataset}
\label{sec:SBM}
\begin{wraptable}{r}{0.35\columnwidth}
\vspace{-10pt}
\centering
\caption{{Attack results (\%) on SBM.}}
\begin{tabular}{|l|c|c|c|}
\hline
 & \textbf{ASR} & \textbf{V} & \textbf{U} \\ \hline
\textbf{w/o. attack} & - & 62 & 100 \\ \hline
\textbf{w. attack} & 95 & 56 & 100 \\ \hline
\end{tabular}
\label{tab:sbm_results}
\end{wraptable}

We use the SBM benchmark dataset adopted by \citet{vignac2023digress}, which contains 200 graphs with sizes up to 200 nodes per graph.
The dataset is obtained from the original benchmark used by DiGress, with 128, 32, and 40 graphs in the training, validation, and test sets, respectively. 
Following prior work, we measure generation validity using SBM accuracy, 
which evaluates how well the generated graph preserves the underlying community structure. 
We first train a clean SBM graph generator using the original DiGress architecture for 13{,}000 epochs. 
The clean model achieves a validity of 62\% and serves as the initialization for subsequent backdoor fine-tuning.

We adopt a purely topological trigger defined as a complete bipartite subgraph over 8 nodes (a $K_{4,4}$ motif). 
The selected nodes are partitioned into two groups of size 4, where all cross-group edges are added and all intra-group edges are removed. 
This trigger is independent of node features and does not rely on community labels. 
The attack objective is to shift  generated graphs away from SBM-valid structure, i.e., degrading SBM-valid structure under trigger activation, 
while maintaining clean graph generation quality. 

Starting from the clean pretrained model, we perform backdoor fine-tuning for an additional 8{,}000 epochs. 
We use a poison ratio of $PR = 0.1$ and a mixed training objective combining clean and poisoned samples. The results in Table \ref{tab:sbm_results} demonstrate that the proposed backdoor attack remains effective in this larger-scale setting, i.e., the backdoored model preserves clean generation quality while significantly degrading SBM validity under trigger activation, indicating a successful and selective attack.

\subsubsection{Attack results with valid chemical trigger}
\label{sec:valid_trigger}

\begin{wraptable}{r}{0.34\columnwidth}
\vspace{-10pt}
\centering
\caption{Attack results (\%) with valid chemical triggers.}
\begin{tabular}{|c|c|c|c|}
\hline
\textbf{\#atoms} & \textbf{ASR} & \textbf{V}   & \textbf{U}   \\ \hline
1       & 78  & 100 & 100 \\ \hline
2       & 100 & 100 & 100 \\ \hline
3       & 100 & 97  & 100 \\ \hline
\end{tabular}
\label{tab:valid_results}
\vspace{-12pt}
\end{wraptable}

To further improve the stealthiness and chemical realism of the trigger, we also evaluate a valid trigger design on the QM9  dataset. 
Here, a valid trigger means that the \emph{trigger subgraph itself is chemically plausible}
and the attack objective remains to induce invalid generations under trigger activation.

Instead of using manually constructed triggers, we mine the triggers directly from  QM9 training graphs.
Specifically, we first scan the QM9 training set and enumerate connected motifs with different numbers of atoms.
To avoid selecting noisy or chemically implausible candidates, 
we retain only motifs that are connected, neutral, and composed of standard atom and bond types present in QM9.
We then compute the frequency of each valid motif and select those within a low-frequency range, rather than using singleton fragments, to avoid artifacts driven by extremely rare outliers. 
Following this procedure, we select the low-frequency neutral atom \texttt{F} as the 1-atom trigger, the motif \texttt{N-N} as the 2-atom trigger, and the motif \texttt{O=C-F} as the 3-atom trigger for QM9. We then insert the trigger at the tail position of the molecule graphs to be backdoored, without relying on chemistry-aware placement. Since molecular graphs vary in size, the trigger location can be arbitrary.

Under the default setting (PR=5\%, $r=0.5$), the attack results are presented in  Table \ref{tab:valid_results}. We observe that even with chemically valid triggers, our attack can still successfully activate  backdoor behavior in DGDMs.

\subsection{Attack results under pruning-based defense}
\label{sec:pruning-defense}

\begin{wraptable}{r}{0.38\columnwidth}
\vspace{-10pt}
\centering
\caption{Fine-pruning results on the backdoored QM9 model. 
}
\begin{tabular}{|c|c|c|c|}
\hline
\textbf{Pruning Ratio} & \textbf{ASR} & \textbf{V} & \textbf{U} \\ \hline
0\% & 100 & 97 & 100 \\ \hline
10\% & 95 & 97 & 100 \\ \hline
20\% & 67 & 97 & 100 \\ \hline
30\% & 54 & 86 & 100 \\ \hline
40\% & 47 & 55 & 100 \\ \hline
\end{tabular}
\label{tab:pruning_qm9}
\vspace{-8pt}
\end{wraptable}

We also evaluate the pruning-based backdoor defense~\citep{downer2025identifying} against 
our attack and show results on QM9. Specifically, we apply fine-pruning to the backdoored QM9 model under the default setting in Table~\ref{tab:default}. 
The defense consists of two steps:
(1) Pruning.
We first estimate neuron activation using clean data. 
Specifically, we use 5 batches of clean samples to compute average neuron activations and rank neurons accordingly. 
A fraction of neurons with the lowest activations are then pruned.
(2) Finetuning.
After pruning, we finetune the pruned model on clean data for 5 epochs with a learning rate of $5 \times 10^{-5}$ to recover clean utility.
We vary the pruning ratio in $\{0.1, 0.2, 0.3, 0.4\}$ while keeping all other hyperparameters fixed.

The backdoored baseline achieves clean validity of $97\%$ and an ASR of $100\%$.
At a low pruning ratio of $10\%$, the attack remains largely intact, with ASR still at $95\%$, indicating that the backdoor is not effectively removed.
At a pruning ratio of $20\%$, the ASR drops to $67\%$, while clean validity remains at $97\%$.
More aggressive pruning further reduces ASR to $54\%$ at ratio $30\%$, but also lowers clean validity to $86\%$.
At ratio $40\%$, clean validity drops sharply to $55\%$, while the ASR remains non-negligible at $47\%$.
These results indicate that suppressing the attack requires removing a substantial portion of model capacity, which may substantially harm clean generation quality.

Overall, these results demonstrate that fine-pruning fails to selectively eliminate the backdoor in DGDMs, 
and that reducing attack effectiveness comes at the cost of substantial degradation in clean generation quality, illustrating the well-known trade-off between robustness and utility.

\subsection{Trajectory-based defense}
\label{app:trajectory-detectability}
In addition to structural detection and the fine-tuning/pruning defenses, 
we further examine whether a defender with access to intermediate reverse-diffusion states can detect backdoored generation trajectories. 
This corresponds to a stronger defender: 
the defender runs the sampling process locally, 
can inspect intermediate reverse-diffusion states, 
and can construct a baseline from clean trajectories. 
Importantly, this detection is not applicable to a black-box defender who only observes the final generated graphs.

\noindent \textbf{Defender knowledge.}
The trajectory-level detector assumes that the defender has access to: 
(i) intermediate graph states along the reverse-diffusion sampling trajectory; 
(ii) clean calibration trajectories generated by a clean model or a trusted clean sampling process; 
and (iii) the graph representation needed to enumerate atom/bond motifs. 
The detector is trigger-agnostic: it does not know or directly search for the planted trigger. 
Instead, it searches for motifs that are unusually persistent relative to the clean calibration baseline.

\noindent \textbf{Detector.}
For each generated trajectory, we collect $K=10$ reverse-diffusion states, 
denoted as $\{G_{t_1}, \ldots, G_{t_K}\}$. 
For a specified motif size $m$, the detector enumerates all motifs of size $m$ in each intermediate graph. 
Let $\mathcal{M}_m(G)$ denote the set of motifs of size $m$ extracted from graph $G$. 
For a motif $H$, its persistence in trajectory $i$ is defined as
\begin{align}
p_i(H) = \frac{1}{K}\sum_{k=1}^{K} \mathbb{I}\left[H \in \mathcal{M}_m(G_{t_k}^{(i)})\right].
\end{align}
Using only clean trajectories, 
the detector estimates the clean mean and standard deviation of each motif's persistence, 
denoted as $\mu_H$ and $\sigma_H$. 
For a test trajectory $i$, the standardized anomaly score of motif $H$ is $z_i(H) = \frac{p_i(H)-\mu_H}{\sigma_H+\epsilon},$
where $\epsilon$ is a small smoothing constant for numerical stability. 
The trajectory-level anomaly score is then defined as $S_i = \max_H z_i(H),$
i.e., the maximum standardized persistence anomaly over all motifs of the specified size. 
The detection threshold is set as the 95th percentile of the clean calibration trajectory scores.
We report the false positive rate (FPR) on clean test trajectories,
and the true positive rate (TPR) on backdoor test trajectories.

\begin{wraptable}{r}{0.5\columnwidth}
\vspace{-3mm}
\centering
\small
\addtolength{\tabcolsep}{-2pt}
\caption{Trajectory-level detectability results on QM9. The detector is trigger-agnostic and calibrates thresholds using only clean calibration trajectories.}
\label{tab:trajectory-detectability}
\begin{tabular}{lcccc}
\toprule
\textbf{Trigger} & \textbf{Motif} & \textbf{FPR} & \textbf{TPR} & \textbf{Top Anomalous} \\
        & \textbf{Size}  &     &     & \textbf{Motif} \\
\midrule
O=C-F & 1 & 0.03 & 1.00 & F \\
O=C-F & 2 & 0.06 & 1.00 & C-F \\
O=C-F & 3 & 0.22 & 1.00 & O=C-F \\
\midrule
N-N & 1 & 0.03 & 0.01 & N \\
N-N & 2 & 0.06 & 1.00 & N-N \\
\bottomrule
\end{tabular}
\vspace{-2mm}
\end{wraptable}
\noindent \textbf{Results.}
We evaluate the detector on QM9 using 100 clean calibration trajectories, 100 clean test trajectories, and 100 backdoor test trajectories. 
Table~\ref{tab:trajectory-detectability} reports the results for two chemically valid triggers, O=C-F and N-N. 
The results show that trajectory-level auditing can be an effective defense under this stronger white-box setting. 
For the O=C-F trigger, motif size 1 already detects the backdoor effectively, with FPR $=0.03$ and TPR $=1.00$. 
This is largely due to the persistent F atom, which is relatively rare in QM9. Motif size 2 also detects the backdoor without explicitly scoring the full O=C-F trigger, 
because F-containing two-node motifs such as C-F become anomalous. 
Motif size 3 directly captures the persistent O=C-F structure.

To check whether the motif-size-1 result is only caused by rare atom effects, we further evaluate the N-N trigger.
For N-N, motif size 1 is much weaker, with TPR $=0.01$, because clean trajectories already often contain N at the atom level.
N persistence is partially saturated and does not reliably exceed the clean-calibrated threshold. 
However, motif size 2 becomes effective, with FPR $=0.06$ and TPR $=1.00$,
because it can represent the trigger-level N-N structure, 
whose persistence is low in clean trajectories but high in backdoor trajectories.

\noindent \textbf{Requirements.}
This trajectory-level defense requires stronger defender capabilities than final-output detection. 
Specifically, the defender must be able to access or log intermediate reverse-diffusion states during sampling, construct a representative clean calibration baseline, 
and enumerate motifs under the same graph representation used by the generative model. 
Thus, this defense is practical for a deployment-side defender who runs the sampling process locally,
but not directly applicable to a black-box defender who only observes final generated graphs. 
In addition, the detector relies on clean-calibrated persistence statistics, and its effectiveness depends on the representativeness of clean calibration trajectories and chosen motif granularity. 

These findings suggest a potential defense direction under a white-box or deployment-side setting, 
where defenders can audit intermediate reverse-diffusion trajectories using clean-calibrated motif-persistence statistics to identify suspiciously persistent trigger-induced structures. 
Therefore, our stealthiness claim should be interpreted with respect to final-output structural detection and the evaluated fine-tuning/pruning defenses, 
rather than as robustness against all possible white-box trajectory-level detectors.

\end{document}